\documentclass[12pt, centerh1]{article}
\usepackage{amssymb,amsmath,amsthm,colonequals,color,multirow,graphicx}
\usepackage{natbib}
\usepackage{setspace}

\theoremstyle{plain}
\newtheorem{thm}{Theorem}[section]
\newtheorem{lem}[thm]{Lemma}

\newtheorem*{cor}{Corollary}

\theoremstyle{definition}
\newtheorem{defn}{Definition}[section]

\theoremstyle{remark}
\newtheorem{rem}{Remark}
\newtheorem*{note}{Note}

\newcommand{\del}{\mathbf\Delta}

\newcommand{\Beta}{\boldsymbol\beta}

\newcommand{\ident}{\mathbf{I}}

\newcommand{\vecx}{\mathbf{x}}
\newcommand{\vecw}{\mathbf{w}}
\newcommand{\matd}{\mathbf{D}}
\newcommand{\mata}{\mathbf{A}}
\newcommand{\vecX}{\mathbf{X}}
\newcommand{\vecU}{\mathbf{U}}

\newcommand{\vecz}{\mathbf{z}}
\newcommand{\vect}{\mathbf{t}}

\newcommand{\vecpi}{\mbox{\boldmath$\pi$}}
\newcommand{\vectheta}{\mbox{\boldmath$\theta$}}

\newcommand{\vecgamma}{\mbox{\boldmath$\gamma$}}
\newcommand{\varthet}{\mbox{\boldmath$\vartheta$}}
\newcommand{\vecmu}{\mbox{\boldmath$\mu$}}
\newcommand{\vecbeta}{\mbox{\boldmath$\beta$}}
\newcommand{\vecalpha}{\mbox{\boldmath$\alpha$}}
\newcommand{\mSigma}{\mbox{\boldmath$\Sigma$}}
\newcommand{\matsig}{\mbox{\boldmath$\Sigma$}}
\newcommand{\mDelta}{\mbox{\boldmath$\Delta$}}

\newcommand{\mPsi}{\mbox{\boldmath$\Psi$}}

\date{}
\begin{document}
\title{A Mixture of Generalized Hyperbolic Distributions}
\author{Ryan P.\ Browne and Paul D.\ McNicholas}
\date{\small Dept.\ of Mathematics \& Statistics, McMaster University, Hamilton, Ontario, Canada.}

\maketitle

\begin{abstract} 
We introduce a mixture of generalized hyperbolic distributions as an alternative to the ubiquitous mixture of Gaussian distributions as well as their near relatives of which the mixture of multivariate $t$ and skew-$t$ distributions are predominant. The mathematical development of our mixture of generalized hyperbolic distributions model relies on its relationship with the generalized inverse Gaussian distribution. The latter is reviewed before our mixture models are presented along with details of the aforesaid reliance. Parameter estimation is outlined within the expectation-maximization framework before the clustering performance of our mixture models is illustrated via applications on simulated and real data. In particular, the ability of our models to recover parameters for data from underlying Gaussian and skew-$t$ distributions is demonstrated. Finally, the role of Generalized hyperbolic mixtures within the wider model-based clustering, classification, and density estimation literature is discussed.\\[+8pt]
\end{abstract}

\section{Introduction} 
Finite mixture models are based on the underlying assumption that a population is a convex combination of a finite number of densities. They therefore lend themselves quite naturally to classification and clustering problems. Formally, a random vector $\mathbf{X}$ arises from a parametric finite mixture distribution if, for all $\vecx \subset \mathbf{X}$, its density can be written $$f(\vecx\mid \mPsi)= \sum_{g=1}^G \pi_g f_g(\vecx\mid\vectheta_g),$$ where $\pi_g >0$, such that $\sum_{g=1}^G \pi_g = 1$, are the mixing proportions, $f_1(\vecx\mid\vectheta_g),\ldots,f_G(\vecx\mid\vectheta_g)$ are called component densities, and $\mPsi=(\vecpi,\vectheta_1,\ldots,\vectheta_G)$ is the vector of parameters with $\vecpi=(\pi_1,\ldots,\pi_G)$. The component densities $f_1(\vecx\mid\vectheta_1),\ldots,f_G(\vecx\mid\vectheta_G)$ are usually taken to be of the same type, most commonly multivariate Gaussian. The popularity of the multivariate Gaussian distribution is due to its mathematical tractability and flexibility for density estimation. In the event that the component densities are multivariate Gaussian, the density of the mixture model is $f(\vecx\mid\mPsi)= \sum_{g=1}^G \pi_g\phi(\vecx\mid\vecmu_g,\matsig_g),$ where $\phi(\vecx\mid\vecmu_g,\matsig_g)$ is the multivariate Gaussian density with mean~$\vecmu_g$ and covariance matrix $\matsig_g$. The idiom `model-based clustering' is used to connote clustering using mixture models. Model-based classification \citep[e.g.,][]{dean06,mcnicholas10c}, or partial classification \citep[cf.][Section~2.7]{mclachlan92}, can be regarded as a semi-supervised version of model-based clustering, while model-based discriminant analysis is supervised \citep[cf.][]{hastie96}.

The recent burgeoning of non-Gaussian approaches to model-based clustering includes work on the multivariate $t$-distribution \citep{peel00,andrews11b,andrews12,lin14}, the skew-normal distribution \citep{lin09}, the skew-$t$ distribution \citep{lin10,vrbik12,vrbik14,lee14,murray14a,murray14b}, the variance-gamma distribution \citep{smcnicholas13}, as well as other approaches \citep{karlis07,handcock07,browne11,franczak14}. In this paper, we add to the richness of the pallet of non-Gaussian mixture model-based approaches to clustering and classification by introducing a mixture of generalized hyperbolic distributions, which contains the aforementioned models as special or limiting cases. Crucially, each component has an index parameter, which is free to vary and engenders a flexibility not present in other non-elliptical mixture approaches to clustering (see Section~\ref{sec:method} for further details).

In Section~\ref{sec:method}, our methodology is developed drawing on connections with the generalized inverse Gaussian distribution. Parameter estimation is described (Section~\ref{sec:para}) before both simulated and real data analyses are used to illustrate our approach (Section~\ref{sec:data}). The paper concludes with a summary and suggestions for future work in Section~\ref{sec:conc}.

\section{Methodology}\label{sec:method}
\subsection{Generalized Inverse Gaussian Distribution}
The generalized inverse Gaussian (GIG) distribution was introduced by \cite{good53} and its statistical properties were laid down by \cite{barndorff77}, \cite{blaesild78}, \cite{halgreen79}, and \cite{jorgensen82}. Write $W\backsim\text{GIG}(\psi,\chi,\lambda)$ to indicate that the random variable~$W$ follows a generalized inverse Gaussian (GIG) distribution with parameters $(\psi,\chi,\lambda)$ and density
\begin{equation}\label{gig}
p(w\mid\psi,\chi,\lambda) = \frac{\left({\psi}/{\chi}\right)^{\lambda/2}y^{\lambda-1}}{2K_\lambda\left( \sqrt{\psi \chi}\right)}\exp\left\{-\frac{\psi y+\chi/y}{2}\right\},
\vspace{-0.08in}
\end{equation}
for $w>0$, where $\psi,\chi\in\mathbb{R}^+$, $\lambda\in\mathbb{R}$, and $K_{\lambda}$ is the modified Bessel function of the third kind with index~$\lambda$. There are several special cases of the GIG distribution, such as the gamma distribution ($\chi=0$, $\lambda>0$) and the inverse Gaussian distribution ($\lambda=-1/2$). Herein, we write $W\backsim\text{GIG}(\psi,\chi,\lambda)$ to indicate that a random variable $W$ has the GIG density as parameterized in \eqref{gig}.

Setting $\chi =  \omega \eta$ and $\psi = {\omega}/{\eta}$ or $\omega = \sqrt{\psi \chi}$ and $\eta = \sqrt{{\chi}/{\psi}}$, we obtain a different but for our purposes, more meaningful parameterization of the GIG density,
\begin{equation} \label{GIG eol}
h(w\mid\omega,\eta,\lambda) = \frac{(y/\eta)^{\lambda-1}}{2\eta K_\lambda\left( \omega \right)}\exp\left\{-\frac{\omega}{2}\left(\frac{y}{\eta}+ \frac{\eta}{y}\right)\right\},
\end{equation}
where $\eta>0$ is a scale parameter, $\omega>0$ is a concentration parameter, and $\lambda$ is an index parameter. Herein, we write $W\backsim\mathcal{I}(\omega,\eta,\lambda)$ to indicate that a random variable $W$ has the GIG density as parameterized in \eqref{GIG eol}.
The GIG distribution has some attractive properties including the tractability of the following expected values:
\begin{equation} \label{GIG scale}
\begin{split}
&\mathbb{E}\left[W\right] = \eta \frac{K_{\lambda+1}\left( \omega \right) }{ K_\lambda\left( \omega\right)},\qquad
\mathbb{E}\left[{1}/{W} \right] = \frac{1}{\eta} \frac{K_{\lambda-1}\left( \omega \right) }{ K_\lambda\left( \omega\right)} = \frac{1}{\eta}\frac{K_{\lambda+1}\left( \omega \right) }{ K_\lambda\left( \omega\right)}  - \frac{2 \lambda}{\omega \eta},\\
&\mathbb{E}\left[\log W\right] = \log \eta + \frac{\partial }{\partial v} \log K_{\lambda}\left(\omega\right) = \log \eta + \frac{1}{ K_{\lambda}\left( \omega\right)}\frac{\partial }{\partial v}   K_{\lambda}\left( \omega\right).
\end{split}\end{equation}
These expected values are necessary ingredients in the expectation-maximization algorithm that is presented in Section~\ref{sec:para}.    

\subsection{Generalized Hyperbolic Distribution}\label{sec:ghd}
\cite{mcneil2005} give the density of a random variable $\vecX$ following the generalized hyperbolic distribution,
\begin{equation}\begin{split} \label{dist ghy1}
f(\vecx\mid\varthet) &= 
\left[ \frac{ \chi + \delta\left(\vecx, \vecmu\mid\del\right) }{ \psi+ \vecalpha'\del^{-1}\vecalpha} \right]^{(\lambda-{p}/{2})/2}\frac{[{\psi}/{\chi}]^{{\lambda}/{2}}K_{\lambda - {p}/{2}}\left(\sqrt{[\psi+ \vecalpha'\del^{-1}\vecalpha] [ \chi +\delta(\vecx, \vecmu\mid\del)]}\right)}{ \left(2\pi\right)^{{p}/{2}} \left|\del\right|^{{1}/{2}} K_{\lambda}\left( \sqrt{\chi\psi}\right)\exp\left\{ \left( \vecmu- \vecx\right)'\del^{-1}\vecalpha\right\}},
\end{split}\end{equation}
where $\delta\left(\vecx, \vecmu\mid\del\right) = \left(\vecx - \vecmu\right)'\del^{-1}\left(\vecx - \vecmu\right)$  is the squared Mahalanobis distance between $\vecx$ and $\vecmu$, and $\varthet=\left(\lambda, \chi, \psi, \vecmu, \mDelta, \vecalpha\right)$ is the vector of parameters. Herein, we use the notation $\vecX\backsim\mathcal{G}_p\left(\lambda, \chi, \psi, \vecmu, \del, \alpha\right)$ to indicate that a $p$-dimensional random variable $\vecX$ has the generalized hyperbolic density in \eqref{dist ghy1}. Note that we use $\del$ to denote the scale because, in this parameterization, we need to hold $|\del|=1$ to ensure idenitifiability (cf.\ Sections~\ref{sec:com} and~\ref{sec:fmi}). 

A generalized hyperbolic random variable $\vecX$ can be generated by combining a random variable $W\backsim\text{GIG}(\psi,\chi,\lambda)$ and a latent multivariate Gaussian random variable $\vecU\backsim\mathcal{N}(\mathbf{0},\del)$ using the relationship
\begin{equation} \label{dist of X given w}
\vecX = \vecmu+ W \vecalpha + \sqrt{W} \vecU,
\end{equation} 
and it follows that $\vecX\mid w \backsim \mathcal{N}(\vecmu+w\vecalpha,w\del)$.
Therefore, from Bayes' theorem,
\begin{equation*}\begin{split}
&f(w\mid\vecx)=\frac{f(\vecx\mid w)h(w)}{f(\vecx)}=\\
&\left[\frac{\psi+\vecalpha'\del^{-1}\vecalpha}{\chi+\delta\left(\vecx,\vecmu\mid\del\right)}\right]^{(\lambda - {p}/{2})/2}
\frac{w^{\lambda-{p}/{2}-1}\exp\left\{-\left[w\left( \psi+ \vecalpha'\del^{-1}\vecalpha\right) + 
\left(\chi+\delta\left(\vecx,\vecmu\mid\del\right)\right)/w\right]/2\right\}}{2K_{\lambda-{p}/{2}} \left( \sqrt{ \left[ \psi+ \vecalpha'\del^{-1}\vecalpha\right] \left[ \chi +\delta\left(\vecx,\vecmu\mid\del\right)\right]}\right)},
\end{split}\end{equation*}
and so we have
$W\mid\vecx\backsim\text{GIG}(\psi+\vecalpha'\del^{-1}\vecalpha, \chi+\delta\left(\vecx,\vecmu\mid\del\right), \lambda-{p}/{2})$.

\cite{mcneil2005} describe a variety of limiting cases for the generalized hyperbolic distribution.
For $\lambda =1$, we obtain the multivariate generalized hyperbolic distribution such that its univariate margins are one-dimensional hyperbolic distributions, for $\lambda = (p+1)/2$, we obtain the $p$-dimensional hyperbolic distribution, and for $\lambda = -{1}/{2}$, we obtain the inverse Gaussian distribution. 
If $\lambda >0$ and $\chi \rightarrow 0$, we obtain a limiting case of the distribution known as the generalized, Bessel function or variance-gamma distribution \citep{barndorff78}. If $\lambda =1 $, $\psi= 2$ and $\chi\rightarrow 0$, then we obtain the asymmetric Laplace distribution \citep[cf.][]{kotz2001} and if $\vecalpha = \mathbf{0}$, we have the symmetric generalized hyperbolic distribution \citep{barndorff78}. Other special and limiting cases include the multivariate normal-inverse Gaussian (MNIG) distribution \citep{karlis07}, the skew-$t$ distribution as well as the multivariate $t$, skew-normal, and Gaussian distributions. 

Suppose we relax the condition that $|\del|=1$, in which case we use $\matsig$ to denote the scale matrix. An identifiability issue arises because the density of $\vecX_1\backsim\mathcal{G}_p(\lambda, \chi/c, c \psi, \vecmu, c \matsig, c \vecalpha) $ and $\vecX_2\backsim\mathcal{G}_p(\lambda, \chi, \psi, \vecmu, \matsig, \alpha)$ is identical for any $c\in\mathbb{R}^+$.  Using $\del$, with $|\del| =1$, instead of $\matsig$, solves this problem but would be prohibitively restrictive for model-based clustering and classification applications. 
An alternative approach is to use the relationship in (\ref{dist of X given w}) to set the scale parameter $\eta=1$. 
This relationship is equivalent to 
$$\vecX = \vecmu + W\eta  \vecalpha + \sqrt{W\eta} \mathbf{U} = \vecmu + W\vecbeta+ \sqrt{W} \mathbf{U},$$
where $\vecbeta= \eta \vecalpha$, 
$W\backsim\mathcal{I}(\omega,1, \lambda)$ and $\vecU\backsim\mathcal{N}(\mathbf{0}, \matsig)$. 
Under this parameterization, the density of the generalized hyperbolic distribution is 
\begin{equation} \label{dist ghy}
f(\vecx\mid\vectheta ) = 
\left[ \frac{ \omega + \delta\left(\vecx, \vecmu| \mSigma\right) }{ \omega+ \vecbeta'\mSigma^{-1}\vecbeta} \right]^{(\lambda-{p}/{2})/2}
\frac{  K_{\lambda - {p}/{2}}\Big(\sqrt{\big[ \omega+ \vecbeta'\mSigma^{-1}\vecbeta\big]\big[\omega +\delta\left(\vecx, \vecmu| \mSigma\right)\big]}\Big)}{ \left(2\pi\right)^{{p}/{2}} \left| \mSigma \right|^{{1}/{2}} K_{\lambda}\left( \omega\right)\exp\big\{-\left(\vecx-\vecmu\right)'\mSigma^{-1}\vecbeta\big\}},
\end{equation}
and $W\mid\vecx\backsim\mathcal{I}(\omega+ \vecbeta'\mSigma^{-1}\vecbeta,\omega+\delta\left(\vecx, \vecmu| \mSigma\right),\lambda - {p}/{2})$. We use $\mathcal{G}_p^*(\lambda,\omega,\vecmu,\mSigma,\vecbeta)$ to denote the density in \eqref{dist ghy} and we use this parameterization when we describe parameter estimation (Section~\ref{sec:para}). 

\subsection{Comments}\label{sec:com}

The presence of the index parameter $\lambda$ in the generalized hyperbolic density engenders a flexibility that is not found in its special and limiting cases. As an illustration of this point, consider the log-densities for three of these special and limiting cases as well as the generalized hyperbolic distributions for two different values of $\lambda$. Looking at Figure~\ref{fig:logd}, it is striking that the Gaussian, variance-gamma, and $t$-distributions have relatively similar log-densities, whereas the generalized hyperbolic distributions have a markedly different log-density. This illustrates the extra modelling flexibility induced by the index parameter~$\lambda$. Furthermore, with appropriate parameterizations, the generalized hyperbolic distribution can replicate the log-density of any of the other distributions in Figure~\ref{fig:logd}. Illustrations of the generalized hyperbolic distribution capturing special and limiting cases are presented in Section~\ref{sec:pararecovery}. Note that the same location and scale parameters were used for each distribution, and, where relevant, skewness was set to $\mathbf{0}$.
\begin{figure}[h!]
\vspace{-0.35in}
\centering
\includegraphics[width=3.5in]{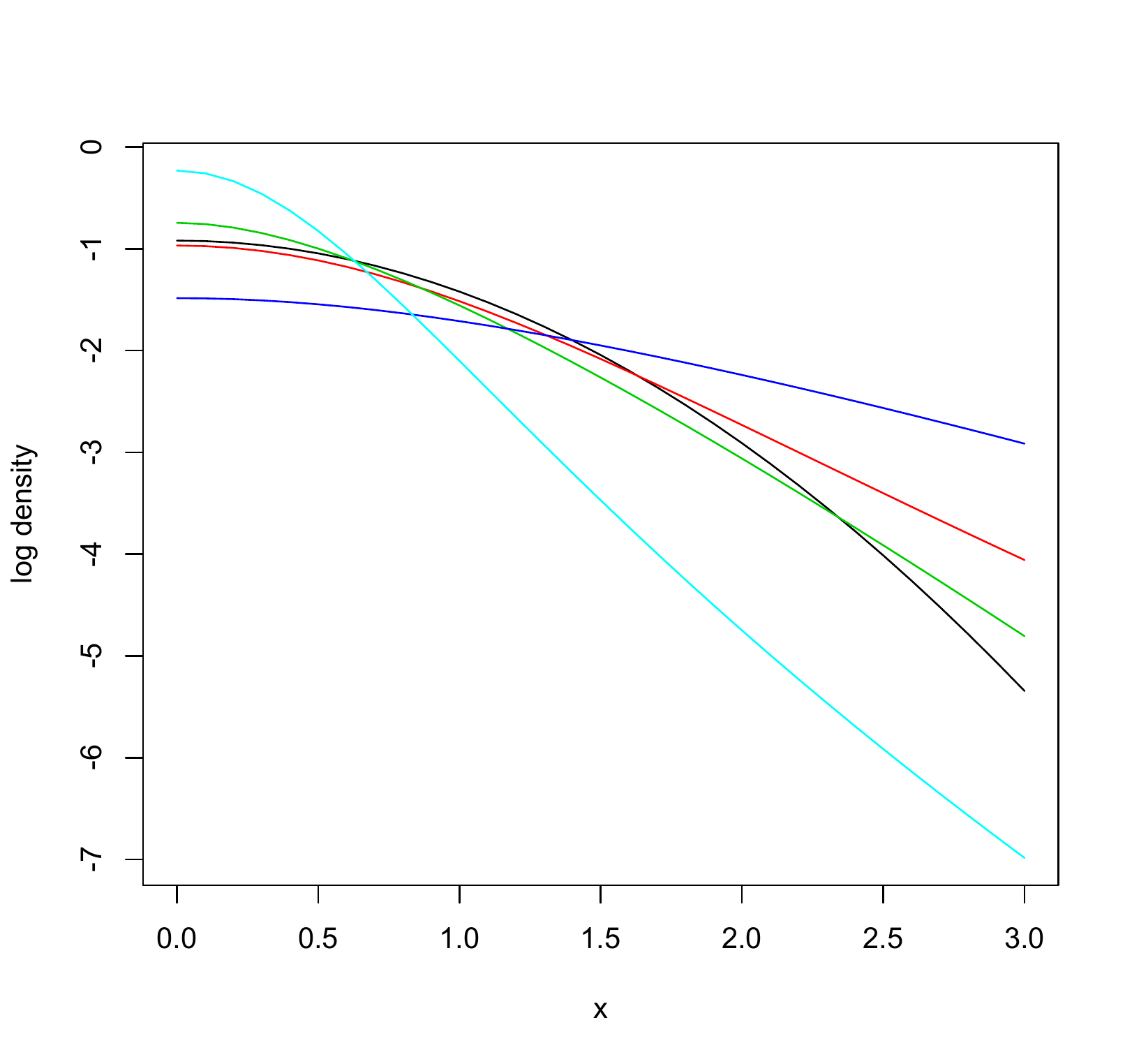}
\vspace{-0.2in}
\caption{Log-density plots for the Gaussian distribution (black), $t$-distribution with five degrees of freedom (red), variance-gamma distribution with five degrees of freedom (green), and the generalized hyperbolic distribution with $\lambda=2$ (blue) and $\lambda=-2$ (turquoise), respectively.}\label{fig:logd}
\end{figure}

\subsection{Identifiability}\label{sec:fmi}

\cite{holzmann2006} prove identifiability of finite mixtures of elliptical distributions. They state that ``finite mixtures are said to be identifiable if distinct mixing distributions with finite support correspond to distinct mixtures''. We will use their work as the basis of a proof that a mixture of generalized hyperbolic distributions is identifiable.

\begin{defn}\label{def:1}
In the present context, finite mixtures from the normal mean-variance (scatter) family ${f_{\alpha_g, p}(x)  : \theta_g =(\theta,\mu, \beta,\sigma) \in  \mathcal{A}^p}$, where $\mathcal{A}^p \subset \mathbb{R}^{k\times p^2 \times p(p+1)/2}$, are identifiable if the following equation is satisfied:
\begin{equation} \label{def identifiable}
\sum_{g=1}^G \pi_g f_{\alpha_g, p}(x) = \sum_{g=1}^G \pi_g' f_{\alpha_g', p}(x),
\end{equation}
$x \in \mathbb{R}^p$, where $G$ is a positive integer, $\sum_{g=1}^G \pi_g = \sum_{g=1}^G \pi_g' = 1$ and $\pi_g, \pi_g' > 0$ for $g=1,\ldots,G$, implies that there exists a permutation $\sigma \in S_m$ such that $(\pi_g,\alpha_g)=(\pi_{\sigma(g)}, \alpha_{\sigma(g)})$ for all $g$. 
\end{defn}

\begin{rem}
Evidently, finite mixtures are identifiable if the family ${f_{\alpha_g, p}(x)  : \theta_g =(\theta,\mu, \beta,\sigma) \in  \mathcal{A}^p}$ is linearly independent. The founding work on finite mixture identifiability is by \cite{yakowitz1968}, who state that this linear independence is a necessary and sufficient condition for identifiability. 
\end{rem}

\begin{defn}Two sets of distributions $\mathcal{G} $ and  $\mathcal{F} $ are disjoint if $\mathcal{G}  \cap\mathcal{F} = \varnothing$ or, equivalently, if no element of $\mathcal{G} $ can be formed by  a linear combination of elements of $\mathcal{F}$, or no element of $\mathcal{F} $ can be formed by  a linear combination of elements of $\mathcal{G}$. That is, the span of $\mathcal{G}$ and  $\mathcal{F}$ is disjoint.\end{defn}

\begin{rem}
If $\mathcal{G}  \cap\mathcal{F} \neq \varnothing $ then there exist some finite $\xi_i$'s or $\tau_j$'s such that: 
\begin{equation}
f(x) = \sum_{i=1}^l \xi_i g_i(x) 
\;\;\;\;\;
\mbox{or}
\;\;\;\;\;
g(x) = \sum_{j=1}^m \tau_j f_j(x),
\end{equation}
where $g \in \mathcal{G}$ and $f \in \mathcal{F}$ 
\end{rem}

\begin{defn}Two sets of distributions $\mathcal{G} $ and  $\mathcal{F} $ are identifiably disjoint if each set is identifiable and the two sets are disjoint.\end{defn}

\begin{defn}A collection of $k$ identifiable sets  of distributions $\left\{ \mathcal{G}_1, \ldots,  \mathcal{G}_k \right\}$ are identifiably mutually disjoint if each set  $\mathcal{G}_i$ is identifiable and each pair $(\mathcal{G}_i, \mathcal{G}_j)$ is disjoint, i.e,  $\mathcal{G}_i \cap \mathcal{G}_j = \varnothing$ for  $i\neq j \in \left\{1,\ldots,k\right\}$.\end{defn}

\begin{lem} If two identifiable sets of distributions  $\mathcal{G} $ and  $\mathcal{F} $ are disjoint then the set $\mathcal{G}  \cup \mathcal{F} $ is identifiable.\end{lem}
\begin{proof} Assume the following relation holds
\begin{equation} \label{union identifiable1}
\sum_{i=1}^l \xi_i g_i(x) + \sum_{j=1}^m \tau_j f_j(x) = 0.
\end{equation}
Both $l \ge 1$ and $m \ge1$ because otherwise we would have $\sum_{i=1}^l \xi_i g_i(x)  = 0$  or $\sum_{j=1}^m \tau_j f_j(x) = 0$, which is a contraction because $\mathcal{G} $ and  $\mathcal{F} $ are identifiable. \eqref{union identifiable1} implies that a relation such as
\begin{equation} \label{union identifiable}
\sum_{i=1}^l \xi_i g_i(x) = - \sum_{j=1}^m \tau_j f_j(x)
\end{equation}
holds and this implies that some linear combination of $g_i(x)$'s lies in the span of $\mathcal{F}$, which is a contraction because $\mathcal{G}  \cap\mathcal{F} = \varnothing $.\end{proof}

\begin{rem}The two sets could be two different distributions, e.g., Cauchy and normal.\end{rem}

\begin{rem}The two sets could divide the parameter space:  $\mathcal{G} =  \left\{ h(x|\theta); \theta \in \Theta_\mathcal{G}\right\}$  and $\mathcal{F} =  \left\{ h(x|\theta); \theta \in \Theta_\mathcal{F} \right\}$. If this forms a partition, i.e., if the complete parameter space is $\Theta =   \Theta_\mathcal{F}  \cup  \Theta_\mathcal{G}$, then we can show identifiability for the complete parameter space by showing identifiability of disjoint pieces of the parameter space by using the standard theorems from \cite{yakowitz1968}, \cite{kent1983}, \cite{holzmann2006}, and \cite{atienza2006}. 
\end{rem}

\begin{rem} This can be generalized to $k$ sets. If a collection of $k$ identifiable sets  $\left\{ \mathcal{G}_1, \ldots,  \mathcal{G}_k \right\}$ are mutually disjoint then the set $\mathcal{F}  =  \cup_{i=1}^k \mathcal{G}_i $ is identifiable.\end{rem}

\begin{rem}This can be further extended if the $k$ sets are indexed by a parameter set, say $\gamma$. Specifically, given a collection of $k$ identifiable sets  $\left\{ \mathcal{G}_{\gamma_1}, \ldots,  \mathcal{G}_{\gamma_k} \right\}$ with index parameter $\gamma$, the set $\mathcal{F}  =  \cup_{i=1}^k \mathcal{G}_{\gamma_i} $ is identifiable if $\mathcal{G}_{\gamma_i}  \cap \mathcal{G}_{\gamma_j} = \varnothing$ for $i\neq j\in \left\{1,\ldots,k\right\}$ and $\gamma \in \Theta$.\end{rem}

\begin{note}In addition, if each $\gamma_i \in \Theta_\gamma $, we can establish the identifiability of $\gamma$ by showing the identifiability of $\sum_{i=1}^k \pi_i g(x | \gamma_i)$, where $\gamma_i \neq \gamma_j$ for $i,j \in \left\{1,\ldots,k\right\}$, using the identifiability theorems.\end{note}

\begin{rem} There is no need to extend these notions of identifiability from finite additivity to $\sigma$-additivity because we are dealing with finite mixtures and, therefore, all subsets will be finite. However, this might be of interest in future research.\end{rem}

\begin{lem}The union, $\mathcal{F}$, of identifiable and mutually disjoint distribution sets $$\left\{ \mathcal{G}_{\gamma,\eta_{\gamma} } | \gamma \in \Theta, \eta_{\gamma} \in \Omega_{\gamma} \right\}$$ is identifiable with respect to $\gamma$ if there is exists a total ordering $\preceq$ on $\gamma \in \Theta$.\end{lem}
\begin{proof}$\mathcal{F}$ is identifiable because, for any finite mixture, there can at most be a finite number, say $k$, of $ \mathcal{G}_\gamma $, labeled as $\gamma_1, \ldots, \gamma_k$, and the set  $\left\{\mathcal{G}_{\gamma_1}, \dots, \mathcal{G}_{\gamma_k}\right\}$ is identifiable and mutually disjoint. 
From the definition of identifiability, i.e., Definition~\ref{def:1}, there exists a relation such that
\begin{equation}
\sum_{i=1}^N \tau_{i} g_{\gamma_{i},  \eta_{\gamma_i} }(x) 
= \sum_{i=1}^N \tau_{i}' g_{\gamma_{i}', \eta_{\gamma_i}'}(x),
\end{equation}
where each $g_{i }(x) \in \mathcal{F}$ and $(\gamma_{i},  \eta_{\gamma_i})$ are the parameters of $\mathcal{F}$. Now, we can order this summation using the total ordering on $\gamma$ and assuming $\gamma_1 \preceq \dots \preceq \gamma_k$, to obtain
\begin{equation}
\sum_{i=1}^k \sum_{j=1}^{m_i} \tau_{ij} g_{\gamma_{i} \theta_{ij} }(x) 
= \sum_{i=1}^{k} \sum_{j=1}^{m_i} \tau_{ij}' g_{\gamma_{i}', \theta_{ij}' }(x)
\end{equation}
for each inner summation. From Definition~\ref{def:1}, there exists a permutation such that $\theta_{ij} = \theta_{ij}'$ and, therefore, there exists a permutation for $\gamma$ as well.\end{proof}   

\begin{cor}For each $\gamma$, the set $\mathcal{G}_{\gamma,\eta_{\gamma} }$ is  identifiable for  $\eta_{\gamma} \in \Omega_{\gamma}$. Each $\eta_{\gamma}$ could represent a different parameter set and $\Omega_{\gamma}$ represents its region of identifiabitily. However, in most cases, the parameters will be equal, i.e., $\eta_{\gamma} = \eta$, and $\Omega_{\gamma_1},\ldots,\Omega_{\gamma_k}$ will be equal in some common region.\end{cor}
Similarly to Lemma 1 from \cite{holzmann2006}, mixture identifiability of multivariate normal mean-variance densities will depend on the identifiability of its univariate analog. With the following two theorems and a corollary, we will prove identifiability of a mixture of generalized hyperbolic distributions.

\begin{thm}\label{thm:little_one}
Recall that the stochastic relationship of the normal mean-variance mixture is given by $$\vecX =\vecmu + W\vecbeta+ \sqrt{W} \mathbf{U},$$ where $\mathbf{U} \backsim \textrm{MVN}_p(\mathbf{0}, \mSigma) $. Let $\vecX$ and $W$ have densities $f_{\theta,p}(x)$ and $f_{\gamma}(w)$, respectively, and suppose that the density $f_{\theta,p}(x)$ is valid up to dimension $q$. If finite mixtures from the univariate family $\mathcal{F}_1 = \left\{ f_{\theta,1}(x) | \theta \in \Theta_1 \right\}$  are identifiable, then finite mixtures from the family $\mathcal{F}_p = \left\{ f_{\theta,p}(x) | \theta \in \Theta_p \right\} $ are also identifiable for $1 \le p <q$.\end{thm}
\begin{proof}
See Appendix~\ref{app:proof0}.
\end{proof}

\begin{thm}\label{thm:big_one} The univariate generalized hyperbolic distribution is identifiable.\end{thm}
\begin{proof} See Appendix~\ref{app:proof}.\end{proof}

\begin{cor} The multivariate generalized hyperbolic distribution is identifiable.\end{cor}
\begin{proof} 
Result follows from Theorems~\ref{thm:little_one} and~\ref{thm:big_one}.\\
\end{proof}

For completeness, the well-known label switching problem should be mentioned when discussing identifiability. \cite{render84} use the term label switching in reference to the ``the invariance of the likelihood under relabelling of the mixture components" \citep{stephens00b}. As \cite{stephens00b} points out, label switching can lead to difficulties when model-based clustering is carried out within the Bayesian paradigm. Herein, we work within the maximum likelihood context, where label switching has no practical implications, arising only as a theoretical identifiability issue that can usually be resolved by specifying some ordering on the mixing proportions, e.g., $\pi_1>\pi_2>\cdots>\pi_G$. Note that in cases where mixing proportions are equal, we could consider a total ordering on other model parameters. In the analyses herein (Section~\ref{sec:data}), we do not necessarily specify such an ordering on the mixing proportions, instead preferring to label the components so as to make the associated tables as reader-friendly as possible.

Note also that, to avoid degenerate solutions arising due to the unboundedness of the likelihood, we choose the root associated with largest local maximum. Other approaches to avoid spurious solutions involve employing constraints or penalty functions. \cite{hathaway85} suggests constraining the ratio of the smallest and largest variance parameters among the components when dealing with a univariate Gaussian distribution. \cite{chen09} use a penalized maximum likelihood estimator and show that it is strongly consistent when the number of components has a known upper bound in the multivariate Gaussian case.

\section{Parameter Estimation}\label{sec:para}

Parameter estimation is carried out using an expectation-maximization (EM) algorithm \citep{dempster77}. The EM algorithm is an iterative technique that facilitates maximum likelihood estimation when data are incomplete or treated as being incomplete. In our case, the missing data comprise the group memberships and the latent variable.  We assume a clustering paradigm so that none of the group membership labels are known. Denote group memberships by $z_{ig}$, where $z_{ig}=1$ if observation $i$ is in component~$g$ and $z_{ig}=0$ otherwise. The latent variables $W_{ig}$ ($i=1,\dots, n$; $g=1,\ldots,G$) are assumed to follow GIG distributions and the complete-data log-likelihood is given by
\begin{equation*}\begin{split} \label{mixture likelihood}
&l_c(\varthet\mid \vecx, \vecw,\mathbf{z}) 
= \sum_{i=1}^n\sum_{g=1}^G z_{ig}\Bigg[\log\pi_g + \sum_{j=1}^p \log \phi\left(\vecx_i\mid\vecmu_g + w_{ig} \vecbeta_g , w_{ig} \mSigma_g \right)  + \log h(w_{ig}\mid\omega_g, \lambda_g )\Bigg]\\ 
&= C -\frac{1}{2}\sum_{i=1}^n\sum_{g=1}^Gz_{ig}\log\left|\mSigma_g^{-1}\right|+\sum_{i=1}^n\sum_{g=1}^Gz_{ig}\log h(w_{ig}\mid\omega_g, \lambda_g)\\ 
& \ -\frac{1}{2}\mbox{tr}\Bigg\{\sum_{g=1}^G \mSigma_g^{-1} \sum_{i=1}^n z_{ig}\bigg[ \frac{1}{w_{ig}}(\vecx_i-\vecmu_g)(\vecx_i-\vecmu_g)'-(\vecx_i-\vecmu_g)\vecbeta_g'-\vecbeta_g(\vecx_i-\vecmu_g)' + w_{ig}\vecbeta_g\vecbeta_g'\bigg]\Bigg\},
\end{split}\end{equation*}
where $C$ does not depend on the model parameters.

In the E-step, the expected value of the complete-data log-likelihood is computed. Because our model is from the exponential family, this is equivalent to replacing the sufficient statistics of the missing data by their expected values in $l_c(\varthet\mid \vecx, \vecw,\mathbf{z})$; here, the missing data are the latent variables and the group membership labels. 
Now, it suffices to calculate the marginal conditional distribution for group memberships given the observed data and for the latent variable in a specific component given the observed data. We require following expectations:
\begin{equation*}\begin{split}
&\mathbb{E}\left[ Z_{ig} \mid \mathbf{x}_i \right] =  \frac{ \pi_g f(\mathbf{x}_i \mid\boldsymbol{\theta}_g)}{\sum_{h=1}^G \pi_h f(\mathbf{x}_i\mid\boldsymbol{\theta}_h)} \equalscolon \hat{z}_{ig},\qquad
\mathbb{E}\left[ W_{ig} \mid \vecx_i, Z_{ig}=1 \right] = \eta_g \frac{K_{\lambda_g+1} (\omega_g)}{K_{\lambda_g} (\omega_g)} \equalscolon a_{ig},\\
&\mathbb{E}\left[{1}/{W_{ig}} \mid \vecx_i, Z_{ig}=1 \right] = \frac{1}{\eta_g} \frac{K_{\lambda_g+1} (\omega_g)}{K_{\lambda_g} (\omega_g)} - \frac{2\lambda_g}{\omega_g \eta_g} \equalscolon b_{ig},\\
&\mathbb{E}\left[ \log(W_{ig}) \mid \vecx_i, Z_{ig}=1 \right] =  \log \eta_g + \frac{1}{K_{\lambda_g}\left(\omega_g\right)}\frac{\partial}{\partial\lambda_g}K_{\lambda_g}\left( \omega_g\right)\equalscolon c_{ig}.
\end{split}\end{equation*}
Hereafter, we use the notation: $n_g= \sum_{i=1}^n\hat{z}_{ig}$, $\bar{a}_g = ({1}/{n_g})\sum_{i=1}^n\hat{z}_{ig} a_{ig}$, $\bar{b}_g = ({1}/{n_g})\sum_{i=1}^n\hat{z}_{ig} b_{ig}$, and $\bar{c}_g = ({1}/{n_g})\sum_{i=1}^n\hat{z}_{ig} c_{ig}$.

In the M-step, we maximize the expected value of the complete-data log-likelihood to get the updates for the parameter estimates. The update for the mixing proportions is $\hat{\pi}_g=n_g/n$, where $n_g=\sum_{i=1}^n\hat{z}_{ig}$. Updates for $\vecmu_g$ and $\vecbeta_g$ are given by
\begin{equation*}\begin{split}
\hat{\vecmu}_g = \frac{ \sum_{i=1}^n \hat{z}_{ig}\vecx_i( \bar{a}_g b_{ig}-1)}{\sum_{i=1}^n \hat{z}_{ig} (\bar{a}_gb_{ig}-1)} \quad\text{and}\quad
\hat{\vecbeta}_g = \frac{ \sum_{i=1}^n\hat{z}_{ig}\vecx_i(\bar{b}_{g}- b_{ig})}{\sum_{i=1}^n \hat{z}_{ig} (\bar{a}_g b_{ig}-1)},
\end{split}\end{equation*}
respectively. The update for $\matsig_g$ is given by
\begin{equation} \label{update sigma}
\hat{\matsig}_g =\frac{1}{n_g}\sum_{i=1}^n\hat{z}_{ig} b_{ig} (\vecx_i - \hat{\vecmu}_g)  (\vecx_i - \hat{\vecmu}_g)'-\hat{\vecbeta}_g\left( \bar{\vecx}_g- \hat{\vecmu}_g \right)'  - \left( \bar{\vecx}_g- \hat{\vecmu}_g \right) (\hat{\vecbeta}_g)' + \bar{a}_g \hat{\vecbeta}_g (\hat{\vecbeta}_g)',
\end{equation}
where $ \bar{\vecx}_g=({1}/{n_g})\sum_{i=1}^n z_{ig}\vecx_i$. 
To demonstrate that $\hat{\matsig}_g$ is positive-definite, first note that, from Jensen's inequality,
${1}/{\mathbb{E}\left[ W_{ig}\right]}  \le  \mathbb{E}\left[ {1}/{W_{ig}} \right]$ for all $i=1,\ldots,n$. 
It follows that ${1}/{a_{ig}}\le b_{ig}$ and so 
\begin{equation*}
\bar{a}_g=\frac{1}{n_g}\sum_{i=1}\hat{z}_{ig}a_{ig}  \ge \frac{1}{n_g}\sum_{i=1}\frac{\hat{z}_{ig}}{b_{ig}}.
\end{equation*}
Now,  replacing $\bar{a}_g$ with $({1}/{n_g})\sum_{i=1}^n({\hat{z}_{ig}}/{b_{ig}})$ in (\ref{update sigma}), we obtain 
\begin{equation*}
\mSigma^{*}_g  =  \frac{1}{n_g} \sum_{i=1}^n z_{ig} b_{ig}\left(\vecx_i - \hat{\vecmu}_g - \frac{1}{b_{ig}}\hat{\vecbeta}_g\right)\left(\vecx_i - \hat{\vecmu}_g - \frac{1}{b_{ig}}\hat{\vecbeta}_g\right)' 
\end{equation*}
and the inequality
$\hat{\mSigma}_g  \succeq \mSigma^{*}_g \succ 0$
holds, ensuring that $\hat{\mSigma}_g$ is positive-definite.


To update $\omega_g$ and $\lambda_g$, we maximize the function
\begin{equation*}
q_g(\omega_g,\lambda_g  ) = -\log K_{\lambda_g} \left(\omega_g \right) + (\lambda_g-1)\bar{c}_g - \frac{\omega_g}{2} \left( \bar{a}_g + \bar{b}_g \right)
\end{equation*}
via conditional maximization, i.e., we maximize the function with respect to one parameter while holding the other parameter fixed. \cite{baricz2010} shows that $K_\lambda \left(\omega\right) $ is strictly log-convex with respect to $\lambda$ and $\omega$. This implies that $q_g(\omega_g,\lambda_g  )$ is strictly concave. The conditional maximization updates are 
\begin{equation*}
\lambda_g^{(t+1)} 
= \ \bar{c}_g \lambda_g^{(t)}  \left[ \frac{ \partial }{ \partial \lambda }  \log K_\lambda \left(\omega_g^{ (t) }\right) \Big|_{\lambda=\lambda^{ (t) } } \right]^{-1}
\ \ \text{ and }\quad
\omega_g^{(t+1)}  
= \omega_g^{ (t) } -  \frac{ \left. \frac{\partial }{\partial \omega  }  q_g\left( \omega,\lambda_g^{(t+1)} \right) \right|_{\omega = \omega_g^{(t)} } }{ \left. \frac{\partial^2 }{\partial \omega^2  } q_g\left(\omega,\lambda_g^{(t+1)} \right) \right|_{\omega = \omega_g^{(t)}}},
\end{equation*}
respectively, and details on their derivation is given in Appendix~\ref{app:updates}. Note that derivatives of the Bessel function are calculated numerically.

The Aitken acceleration \citep{aitken26} can be used to estimate the asymptotic maximum of the log-likelihood at each iteration of an EM algorithm and thence to determine convergence. The Aitken acceleration at iteration $k$ is 
$$a^{(k)}=\frac{l^{(k+1)}-l^{(k)}}{l^{(k)}-l^{(k-1)}},$$
where $l^{(k)}$ is the log-likelihood at iteration $k$. An asymptotic estimate of the log-likelihood at iteration $k+1$ is
$$l_{\infty}^{(k+1)}=l^{(k)}+\frac{1}{1-a^{(k)}}(l^{(k+1)}-l^{(k)}),$$
and the algorithm can be considered to have converged when
$l_{\infty}^{(k)}-l^{(k)}<\epsilon$, provided that this difference is positive \citep{bohning94,lindsay95,mcnicholas10a}.
This criterion is used for the analyses in Section~\ref{sec:data}, with $\epsilon=0.01$. 

After the algorithm has converged, the predicted classifications are given by the \text{a~posteriori} probabilities (expected values) $\hat{z}_{ig}$. When reporting predicted classification results in our data analyses (Section~\ref{sec:data}), we use maximum \textit{a~posteriori} probabilities (MAP) given by 
$$\text{MAP}(\hat{z}_{ig})=
\begin{cases}
1 & \text{if } \max_h\{\hat{z}_{ih}\} \text{ occurs in component }g,\\
0 & \text{otherwise}.
\end{cases}$$

In many practical applications, the number of mixture components $G$ is unknown. In our illustrative data analyses (Section~\ref{sec:data}), $G$ is treated as unknown and the Bayesian information criterion \citep[BIC;][]{schwarz1978} is used for model selection. The BIC can be written $\text{BIC}=2l(\mathbf{x}~|~\hat{\boldsymbol{\theta}}) - \rho \log n$, where $l(\mathbf{x}~|~\hat{\boldsymbol{\theta}})$ is the maximized log-likelihood, $\hat{\boldsymbol{\theta}}$ is the maximum likelihood estimate of ${\boldsymbol{\theta}}$, and $\rho$ is the number of free parameters. The use of the BIC for mixture model selection can be motivated through Bayes factors \citep{kass95,kass95b,dasgupta98} and has become popular due to its widespread use within the Gaussian mixture modelling literature. While many alternatives have been proffered, none have yet proved superior.

\section{Data analyses}\label{sec:data}

\subsection{Overview}
The mixture of generalized hyperbolic distributions model is illustrated on simulated and real data. We consider cluster analyses, but these mixture models could equally well be applied for semi-supervised classification, discriminant analysis, or density estimation. In each of our clustering analyses, the true classes are known but treated as unknown for illustration. While this sort of synthetic clustering example may not be considered quite akin to real clustering, it is representative of what has become the norm within the model-based clustering literature. Furthermore, the real data sets that we use are selected because of their popularity as benchmark data sets within the aforesaid literature. The crabs data (Section~\ref{sec:crabs}), in particular, are notoriously difficult to cluster. 

Because we know the true classes, we can assess the performance of these mixture models in terms of classification accuracy, which we measure using the adjusted Rand index \citep[ARI;][]{rand71,hubert85}. The ARI has expected value 0 under random classification and takes the value 1 for perfect class agreement. Negative values of the ARI indicate classification that is worse than would be expected under random classification. 

For the real data analyses (Section~\ref{sec:realdata}), we use an \textit{em}EM approach for initialization \citep[cf.][]{biernacki2000}. Specifically,  for each of 100 random starting values, the EM algorithm is run for 50 iterations. The algorithm with the highest (finite) log-likelihood after 50 iterations is then iterated until convergence. The simulated data examples (Section~\ref{sec:pararecovery}) are used to illustrate parameter recovery from the true model only, and they are not difficult clustering problems; accordingly, one $k$-means start suffices for each run. Note that we suggest using an initialization approach with multiple starts, such as the \textit{em}EM approach, in any real application.

\subsection{Simulated data analyses}\label{sec:pararecovery}
Data are simulated to illustrate the effectiveness of our mixture of generalized hyperbolic distributions, with our parameter estimation approach, for modelling data from its special and limiting cases. Data are simulated from a mixture of Gaussian distributions and a mixture of skew-$t$ distributions, respectively, and true and estimated parameters are compared. 
In each case, 100 two-component data sets are simulated with $n_1=n_2=250$ and the models are fitted within the model-based clustering paradigm for $G=1,\ldots,5$. In all cases, a $G=2$ component model is selected, perfect classification results are obtained, and the parameter estimates are close to the true values (Tables~\ref{tab:simG} and~\ref{tab:simst}). Note that the results are reported in the $(\psi_g, \chi_g)$ parameterization for ease of interpretation; however, they were run with the $\omega_g$ parameterization.

It is notable that, in these limiting and special cases, the parameters in question are estimated very well by our mixture of generalized hyperbolic distributions. In the Gaussian mixture simulation (Table~\ref{tab:simG}), both components should and do have large $\chi_g$ as well as small $\psi_g$ and $\vecalpha_g$. For the skew-$t$ mixture example (Table~\ref{tab:simst}), we generated data with degrees of freedom $\chi_1=8, \chi_2=20$ and the average estimates were very close ($\hat\chi_1=8.24, \hat\chi_2=21.02$). Again, $\psi_g$ should be and is small. Note that $\lambda_g$ is not free to vary for either the Gaussian or skew $t$-distributions, i.e., it is constrained so that $\lambda_g=-\chi_g/2$. 
\begin{table}[h!]
\caption{\label{tab:simG}Mean parameter estimates from the application of our mixture of generalized hyperbolic distributions to 100 simulated data sets from a two-component mixture of Gaussian distributions.}
\centering
\begin{tabular*}{1\textwidth}{@{\extracolsep{\fill}}*{6}{r}}
\hline
 & \multicolumn{2}{c}{$g=1$} & & \multicolumn{2}{c}{$g=2$}\\
 \cline{2-3}\cline{5-6}
 & True  & Estimated &  & True & Estimated\\
  &   & (Std.\ Dev.) &  &  & (Std.\ Dev.)\\
\hline
$\vecmu_g$ & $(3.00,3.00)$ & $(2.91,2.97)$ &&$(-3.00,-3.00)$ & $(-2.40,-3.42)$\\
 &  & $(2.37, 2.51)$ &&  & $(2.26,2.29)$\\
$\vecalpha_g$ &$(0.00,0.00)$ & $(-0.03,0.16)$&& $(0.00,0.00)$ &$(-0.26,0.14)$\\
 &  & $(2.37, 2.51)$ &&  & $(2.26,2.29)$\\
$\del_g$ & $(1.51,-1.13,1.51)$ &$(1.51,-1.13,1.52)$ &&$(1.51,-1.13,1.51)$ & $(1.54,-1.18,1.56)$\\
&  & $(0.18,0.17,0.17)$ &&  & $(0.23,0.21,0.21)$\\
$\psi_g$ & $\rightarrow0$& $0.00$&  &$\rightarrow0$ &$0.00$\\
 &  & $(0.0002)$ &&  & $(0.0002)$\\
$\chi_g$ &$\rightarrow\infty$ & $70.3$ &  & $\rightarrow\infty$&$69.1$\\
 &  & $(7.1)$ &&  & $(7.75)$\\
$\lambda_g$ &$\rightarrow-\infty$ & $-119.8$ &  & $\rightarrow-\infty$&$-120.5$\\ 
 &  & $(2.33)$ &&  & $(2.15)$\\ \hline
\end{tabular*}
\end{table}

\begin{table}[h!]
\caption{\label{tab:simst}Mean parameter estimates from the application of our mixture of generalized hyperbolic distributions to 100 simulated data sets from a two-component mixture of skew-$t$ distributions.}
\centering
\begin{tabular*}{\textwidth}{@{\extracolsep{\fill}}*{6}{r}}
\hline
& \multicolumn{2}{c}{$g=1$} & & \multicolumn{2}{c}{$g=2$}\\
 \cline{2-3}\cline{5-6}
 & True  & Estimated &  & True & Estimated\\
  &   & (Std.\ Dev.) &  &  & (Std.\ Dev.)\\
\hline
$\vecmu_g$ & $(3.00,3.00)$ & $(2.95,3.04)$ &&$(-3.00,-3.00)$ & $(-2.89,-3.11)$\\
   &  & $(0.45, 0.45)$ &&  & $(0.30,0.30)$\\
$\vecalpha_g$ &$(2.00,-2.00)$ & $(2.05,-2.06)$&& $(-1.00,1.00)$ &$(-1.30,1.36)$\\
   &  & $(0.53, 0.53)$ &&  & $(0.30,0.30)$\\
$\matsig_g$ & $(1.00,-0.75,1.00)$ &$(0.99,-0.74,0.98)$ &&$(1.00,-0.75,1.00)$ & $(1.01,-0.76,1.01)$\\
   &  & $(0.19, 0.18, .17)$ &&  & $(0.11, 0.11, .11)$\\
$\psi_g$ & $\rightarrow0$& $0.00$&  &$\rightarrow0$ &$0.00$\\
   &  & $(0.001)$ &&  & $(0.001)$\\
$\chi_g$ &$8$ & $8.24$ &  &$20$ &$21.02$\\
   &  & $(2.53)$ &&  & $(1.8)$\\
$\lambda_g$ &$-4$ & $-4.12$ &  &$-10$ &$-10.51$\\
   &  & $(0.83)$ &&  & $(0.68)$\\ \hline
\end{tabular*}
\end{table}

\subsection{Real data analyses}\label{sec:realdata}
\subsubsection{Leptograpus crabs data}\label{sec:crabs}

\cite{campbell74} reported data on five biological measurements on 200 crabs from the genus \textit{leptograpus}. The data were collected in Fremantle, Western Australia and comprise 50 male and 50
female crabs for each of two species: orange and blue. The data are sourced from the {\tt MASS} library for {\sf R} \citep{r2013}, which contains data sets from \cite{venables99}. These data are used by \cite{raftery06} and \cite{andrews14} to illustrate respective variable selection techniques for model-based clustering. 

Mixtures of generalized hyperbolic distributions are fitted to these data for $G=1,\ldots,9$. The BIC chooses a $G=2$ component model and the resulting MAP classifications separate the crabs by colour ($\text{ARI}=0.50$; cf.\ Table~\ref{tab:crabsgh0}). For a Gaussian mixture model fitted over the same range of $G$, the BIC chooses a $G=2$ component model with classification performance akin to guessing ($\text{ARI}=0.03$; cf.\ Table~\ref{tab:crabsgh0}). Note that the {\tt mclust} software \citep{fraley13} for {\sf R} was used to fit this Gaussian mixture model (with unconstrained covariance).
\begin{table}[h!]
\caption{\label{tab:crabsgh0}Colour and gender versus predicted classifications for the chosen generalized hyperbolic (GH) mixture and the chosen Gaussian mixture for the crabs data.}
\centering
\begin{tabular*}{1\textwidth}{@{\extracolsep{\fill}}ll*{7}{r}}
\hline
&&\multicolumn{2}{c}{GH Mixture}&&\multicolumn{2}{c}{Gaussian Mixture}\\
\cline{3-4}\cline{6-7}
 &&  1 & 2 &  &1&2\\ \hline
\multirow{2}{*}{Blue} & Male    &50& 0& &21 &29 \\ 
& Female       &50&0& &26&24\\ 
\multirow{2}{*}{Orange} &Male   & 0 & 50 &  &24&26\\ 
 &Female & 0 & 50 & &9&41 \\ \hline
\end{tabular*}
\end{table}

Applying the famous MCLUST models \citep{fraley02a} to these data, using the {\tt mclust} software for {\sf R}, and using the same range of values for $G$, results in a $G=4$ component model that misclassifies 80 crabs ($\text{ARI}=0.31$; cf.\ Table~\ref{tab:crabsmclust}). While merging Gaussian components can sometimes improve classification performance, it will not help with either the chosen Gaussian mixture model (Table~\ref{tab:crabsgh0}) or the best MCLUST model (Table~\ref{tab:crabsmclust}) for the crabs data. That our generalized hyperbolic mixtures outperform MCLUST here is especially significant when one considers that MCLUST is sometimes considered a sort of gold standard approach for model-based clustering. The MCLUST family of models is a subset of 10 of the 14 Gaussian parsimonious clustering models (GPCMs) used by \cite{celeux95}; see Appendix~\ref{app:gpcm} for details.\begin{table}[h!]
\caption{\label{tab:crabsmclust}Colour and gender versus predicted classifications for the chosen MCLUST model for the crabs data.}
\centering
\begin{tabular*}{1\textwidth}{@{\extracolsep{\fill}}ll*{4}{r}}
\hline
 &&  1 & 2 & 3 & 4\\ \hline
\multirow{2}{*}{Blue} & Male & 19 & 0 & 31 & 0\\ 
& Female & 16 & 33 & 0 & 1\\ 
\multirow{2}{*}{Orange} &Male &25 & 0 &25 & 0\\ 
 &Female &5 &14 & 0 & 31\\\hline
\end{tabular*}
\end{table}

A reviewer also asked us to comment on the MCLUST results reported by \cite{melnykov2013}, who suggests that MCLUST does poorly on the crabs data set because of ``purely an initialization issue, since after the careful selection of starting parameter values for the EM algorithm, a 4-component mixture with unrestricted covariances is preferred according to BIC''.  Following this logic, we use the {\tt mixture} package \citep{browne13} for {\sf R}, using the \textit{em}EM approach for initialization, as suggested by \cite{biernacki2000} and used by \cite{melnykov2013}. In addition to allowing flexibility vis-\`{a}-vis starting values, the {\tt mixture} package implements all 14 GPCMs via algorithms provided by \cite{celeux95} and \cite{browne13b}.
Specifically, we use 100 random starting values for the \textit{em}EM followed by 50 iterations of EM algorithm. For comparison, we take an identical approach to fitting our generalized hyperbolic distribution.

Within this framework, if we consider only the unconstrained model (VVV; cf. Table~\ref{tab:mclust2}, Appendix~\ref{app:gpcm}), the BIC selects a four-component MCLUST model with ARI of 0.82 --- this is the same solution found by \cite{melnykov2013}. However, if all 14 covariance structures are considered, the BIC selects a seven-component EEE model with $\text{ARI}=0.53$ (Table~\ref{tab:crabsgh}).
To facilitate a comprehensive comparison, we develop 14 mixtures of generalized hyperbolic distributions such that their scale matrices correspond to the GPCM covariance matrices (cf.\ Table~\ref{tab:mclust2}, Appendix~\ref{app:gpcm}). Running all 14 models, for $G=1,\ldots,9$, the BIC selects a four-component 
model (with EEE scale matrix) and the resulting MAP classifications leave only 15 crabs misclassified (ARI = 0.82; Table~\ref{tab:crabsgh}). Accordingly, the chosen generalized hyperbolic mixture (EEE, $G=4$, $\text{ARI}=0.53$) outperforms the chosen GPCM model (EEE, $G=4$, $\text{ARI}=0.82$) within this \textit{em}EM initialization framework. 
\begin{table}[h!]
\caption{\label{tab:crabsgh}Colour and gender versus predicted classifications for the chosen generalized hyperbolic (GH) mixture (EEE scale matrix) and the chosen GPCM (EEE).}
\centering
\begin{tabular*}{1\textwidth}{@{\extracolsep{\fill}}ll*{12}{r}}
\hline
&&\multicolumn{4}{c}{GH Mixture (EEE)}&&\multicolumn{7}{c}{GPCM (EEE)}\\
\cline{3-6}\cline{8-14}
 &&  1 & 2 & 3 & 4&&1&2& 3& 4 & 5 &6 & 7\\ \hline
\multirow{2}{*}{Blue} & Male    &39&11  & 0&0 &&18 &   0  &  0   &  0   &   0  &  0   &  32 \\ 
& Female       &0&50&0 &0&&19 &  0   &  0   & 0    & 0    &  31 &  0   \\ 
\multirow{2}{*}{Orange} &Male   & 0 & 0 & 50 &0&& 0 &  28 &  0   &  22 &   0  &  0   &   0  \\ 
 &Female & 0 & 0 &  4&46&  & 0  &  0   &  24 &   5 &  21 &  0   &  0   \\ \hline
\end{tabular*}
\end{table}

In addition to outperforming the Gaussian mixture model and MCLUST, the performance of our mixture of generalized hyperbolic distributions on the crabs data compares favourably with many other analyses throughout the literature. At the request of a reviewer, we also apply the scale mixture of skew-normal distributions approach developed by \cite{cabral2012} and implemented in the {\sf R} package {\tt mixsmsn} \citep{prates2013}. This approach gave poor performance on the crabs data, returning a one-component model.

\subsubsection{Old Faithful Data}

Because of the very poor performance of the scale mixture of skew-normal distributions approach of \cite{cabral2012} on the crabs data, we conduct a second comparison with our generalized hyperbolic mixtures. This time the famous Old Faithful geyser data are used. We choose these data because they are used by \cite{prates2013} to illustrate the scale mixture of skew-normal distributions approach. 

The data comprise the waiting time between and the duration of eruptions of the Old Faithful geyser data in Yellowstone National Park. There are no true classes \textit{per se} for these data; however, there is some consensus that there are two components --- more frequent, shorter eruptions and less frequent, longer eruptions --- and that they are skewed. For both the scale mixture of skew-normal distributions and the generalized hyperbolic approaches, a $G=2$ component model is selected. The classification results are identical; however, contour plots reveal that the fit of the generalized hyperbolic mixture model is vastly superior to the fit of the scale mixture of skew-normal distributions for these data (Figure~\ref{fig:logd2}).
\begin{figure}[h!]
\vspace{-0.25in}
\centering
\hspace*{-0.2in}\includegraphics[width=0.515\textwidth,keepaspectratio=true]{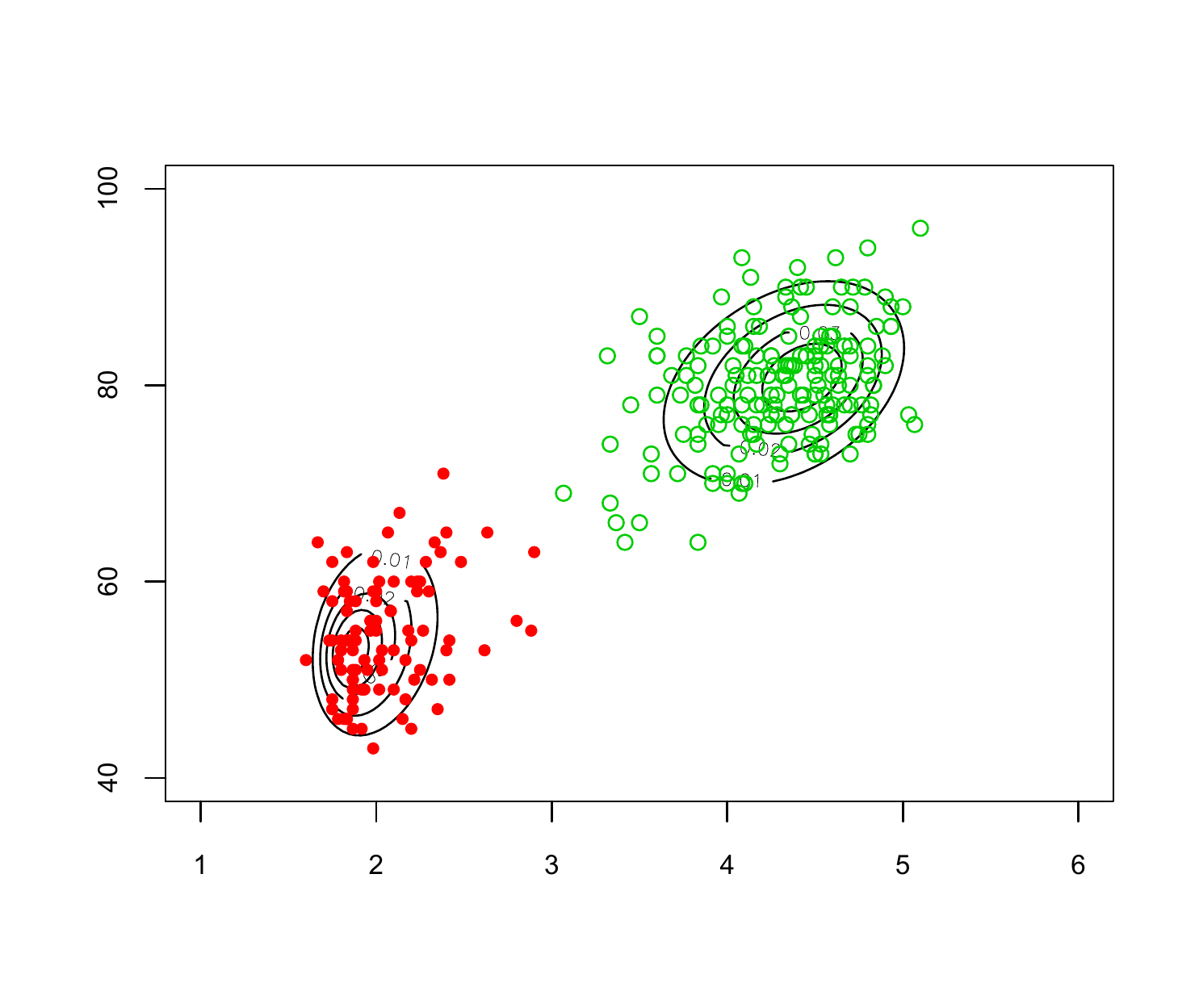}
\hspace*{-0.05in}\includegraphics[width=0.515\textwidth,keepaspectratio=true]{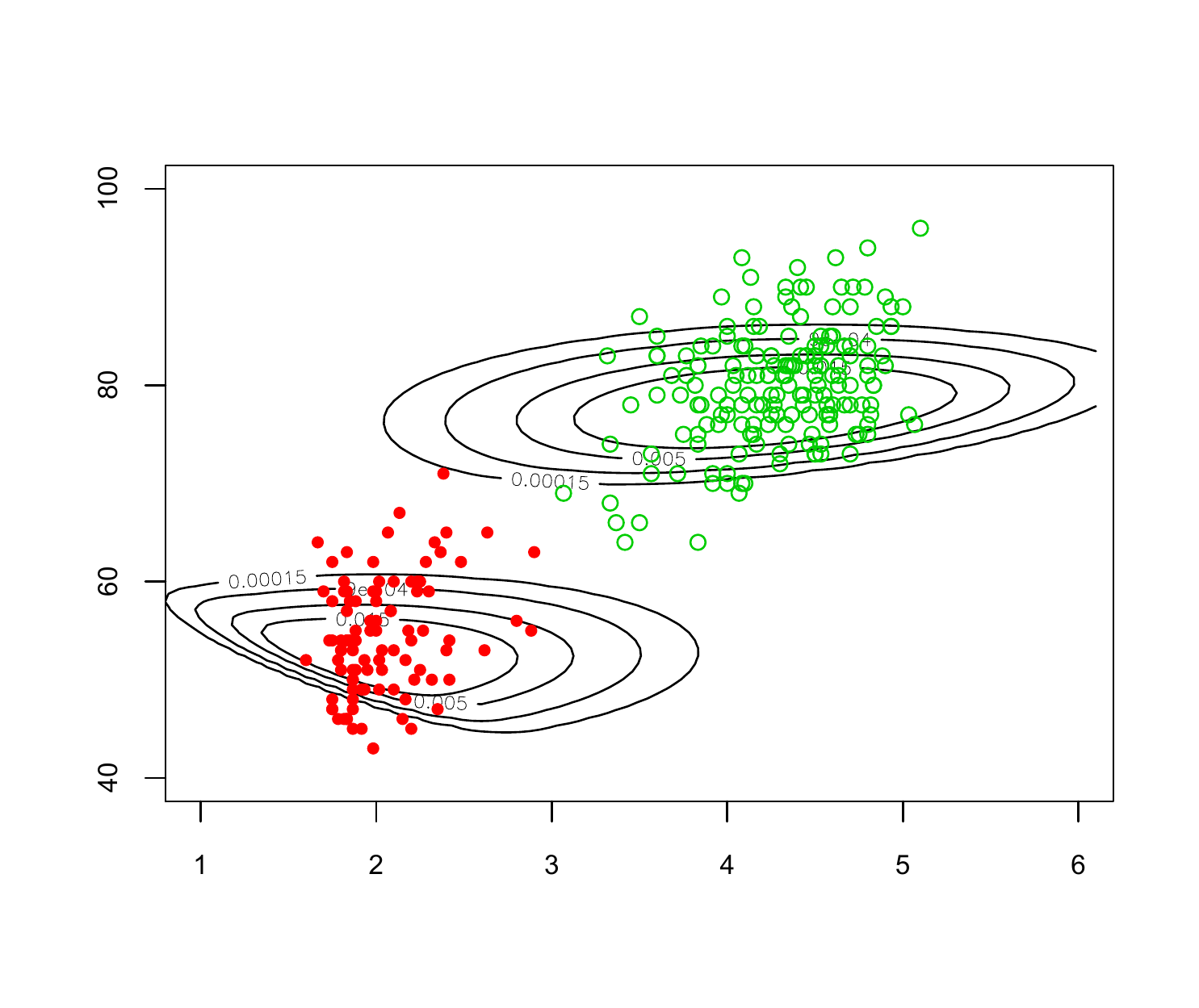}
\vspace{-0.55in}
\caption{Contour plots for the generalized hyperbolic mixtures (left) and scale mixture of skew-normal mixtures (right) for the Old Faithful data, where colour denotes predicted classification.}\label{fig:logd2}
\end{figure}

\section{Discussion}\label{sec:conc}
A mixture of generalized hyperbolic distributions has been introduced. Parameter estimation, via an EM algorithm, was enabled by exploitation of the relationship with the GIG distribution. The mixture models were illustrated via clustering applications on real data, where they performed favourably when compared to a  Gaussian mixture model, the well-established MCLUST/GPCM approach, and the scale mixture of skew-normal distributions approach developed by \cite{cabral2012}. Although illustrated for clustering, mixtures of generalized hyperbolic distributions can also be applied for semi-supervised classification, discriminant analysis, and density estimation. They represent perhaps the most flexible in a recent series of alternatives to the Gaussian mixture model for clustering and classification. What sets the mixture of generalized hyperbolic distributions apart from several other non-elliptical mixture approaches is the presence of an index parameter, and the fact that it has many common models as special or limiting cases. In a simulation, we showed that the generalized hyperbolic mixtures, with our EM algorithm, performs very well when recovering parameters for limiting and special cases.

The precise impact of our generalized hyperbolic mixtures within the wider model-based clustering literature is difficult to predict at present. Certainly, we do not suggest that one should use our mixtures of generalized hyperbolic approaches exclusively, completely ignoring more well-established approaches such as Gaussian mixtures. However, results obtained to date suggest that application of mixtures of generalized hyperbolic distributions in real cluster analyses can outperform its special cases and this should not be ignored. One downside of our generalized hyperbolic mixtures is that it will take longer to run; however, implementing the code in a lower level language and parallel implementation can be used to greatly reduce the runtime, and both are topics of ongoing work.

\appendix
\section{Updates for $\lambda_g$ and $\omega_g$}\label{app:updates}
For this section, we drop the subscript $g$ to ease the notational burden, so we have
\begin{equation*}
q(\omega,\lambda  ) = -\log K_{\lambda}\left(\omega \right) + (\lambda-1)\bar{c} - \frac{\omega}{2} \left( \bar{a} + \bar{b} \right).
\end{equation*}
\cite{baricz2010} shows that $K_\lambda \left(\omega\right) $ is strictly log-convex is with respect to $\lambda$ and $\omega$. Thus, we use conditional maximization to obtain updating equations. For $\omega$, we use Newton's method using the first and second derivative of $q$ with respect to $\omega$, i.e.,  
\begin{equation*}
\frac{ \partial }{ \partial \omega }q(\omega,\lambda) = - \frac{ K_{\lambda}' (\omega) }{ K_{\lambda} (\omega ) } - \frac{1}{2} \left( \bar{a} + \bar{b} \right) = \frac{1}{2} \left[ R_{\lambda} (\omega) + R_{-\lambda} (\omega )  - \left( \bar{a} + \bar{b} \right)  \right],
\end{equation*}
\begin{equation*}
\frac{ \partial^2 }{ \partial \omega^2 }q(\omega,\lambda) =  \frac{1}{2} \left[ R_{\lambda} (\omega)^2  - \frac{1+2\lambda}{\omega}R_{\lambda} (\omega) -1 + R_{-\lambda} (\omega)^2  - \frac{1-2\lambda}{\omega} R_{-\lambda} (\omega) -1   \right],
\end{equation*}
where $R_{\lambda} (\omega) = K_{\lambda+1} (\omega)/ K_{\lambda} (\omega)$.




To update $\lambda$, we construct a surrogate function by deriving a bound on the second derivative via the following integral representation of the modified Bessel function of the third kind \citep[][pg. 181]{watson1944}. The second derivative of $\cosh(t)$ is actually $\cosh(t)$, which is bounded below by 1. Thus, we can construct a quadratic function such that
\begin{eqnarray*}
\cosh(t) & \ge & 1 + t^2 \\
\exp\{- \omega \cosh(t)\} & \le & \exp\{-\omega (1 + t^2) \}  \qquad\quad \mbox{for} \; \omega >0 \\
\exp\{- \omega \cosh(t)\}\cosh( \lambda t) & \le & \exp\{-\omega (1 + t^2) \} \cosh( \lambda t)  \\
\int_0^\infty \exp\{- \omega \cosh(t)\}\cosh( \lambda t) \mbox{dt} & \le & \int_0^\infty \exp\{-\omega (1 + t^2) \} \cosh( \lambda t)   \mbox{dt}\\
K_{\lambda}(\omega) & \le & \sqrt{\frac{\pi}{2\omega}} \exp \left\{  \frac{ \lambda^2 - 2 \omega^2 }{2\omega}  \right\} \\
\log K_{\lambda}(\omega) & \le & \log\sqrt{\frac{\pi}{2\omega}} + \frac{\lambda^2 - 2\omega^2}{2\omega}  
\end{eqnarray*}
Because both these functions are convex in $\lambda \in \mathbb{R}$, the inequality 
\begin{equation*}
0 
\le \frac{\partial^2}{\partial \lambda^2 } \log K_{\lambda}(\omega)
\le  \frac{1}{\omega}
\end{equation*}
holds for all  $\lambda \in \mathbb{R}$. Because the second derivative is bounded and $\log K_{\lambda}(\omega)$ is an even function with respect to $\lambda$, we follow \cite{deleeuw2009} and construct the majorizing surrogate function
\begin{equation*}
g(\lambda | \lambda_0) = \log K_{\lambda_0}(\omega) +  \frac{  \lambda^2 - \lambda_0^2   }{2\lambda_0} \left( \frac{\partial}{\partial \lambda } \log K_{\lambda}(\omega) \Big|_{\lambda=\lambda_0} \right).
\end{equation*}
Applying this surrogate function to $q(\omega,\lambda  )$, we obtain the desired update for $\lambda$. 

\section{Proof of Theorems~\ref{thm:little_one} and~\ref{thm:big_one}}\label{app:proof}
%
\subsection{Proof of Theorem~\ref{thm:little_one}}\label{app:proof0}
\begin{proof}
Suppose there exists a relationship such that
\begin{equation}  \label{ nonidentify } 
\sum_{g=1}^G \tau_g f_{\theta_g,p}(\vecx ) = 0,
\end{equation}
where $\tau_g \in \mathbb{R}$ and where $\theta_g$ is pairwise distinct.  The characteristic function for a normal mean-variance density  is 
\begin{equation} 
\varphi_{X}( \vect) = 
\exp \left\{ i  \vect'\vecmu_g \right\} 
M_{W} \left( \Beta_g' \vect -\frac{1}{2} \vect' \mSigma_g \vect~\Bigg|~\vecgamma_g \right),
\end{equation}
where $M_W(w)$ is the moment generating function for $W$. Applying the Fourier transform, we obtain the characteristic function of the linear combination
\begin{equation}
\sum_{g=1}^G \tau_g \exp \left\{ i  \vect'\vecmu_g \right\} 
M_{W} \left( \Beta_g' \vect -\frac{1}{2} \vect' \mSigma_g \vect~\Bigg|~\vecgamma_g \right) = 0.
\end{equation}
\citet[][p.~211]{yakowitz1968} show that there exists $\vecz \in \mathbb{R}^p$ such that the tuples $\alpha_g^z = (\vecmu_g' \vecz, \vecz' \mSigma_g \vecz )$ are pairwise distinct. A similar argument shows that there exists $\vecz \in \mathbb{R}^p$ such that the tuples $\theta_g^z = ( \vecgamma_g, \Beta_g'\vecz, \vecmu_g' \vecz, \vecz' \mSigma_g \vecz )$ are pairwise distinct. This follows simply by setting $\alpha_{g1}^z = (\vecmu_g' \vecz, \vecz' \mSigma_g\vecz) $ and $\alpha_{g2}^z = (\Beta_g' \vecz, \vecz' \mSigma_g\vecz)$, and then letting $\theta_g^z = (\vecgamma_g, \alpha_{g1}^z, \alpha_{g2}^z)$. 
Now setting $\vect = u \vecz$, for $u\in \mathbb{R}$, we have
\begin{equation}
\sum_{g=1}^G \tau_g\exp \left\{ i  u \vecz' \vecmu_g \right\} 
M_{W} \left( u \Beta_g' \vecz -\frac{u^2}{2} \vecz' \mSigma_g \vecz~\Bigg|~\vecgamma_g \right) = 0,
\end{equation}
for $u \in \mathbb{R}$. Taking the one-dimensional inverse Fourier transform yields 
\begin{equation} 
\sum_{g=1}^G \tau_g f_{\theta_g^z, 1}(\vect ) = 0,
\end{equation}
which implies that $\tau_1 = \dots = \tau_G = 0$ by the identifiability assumption on $\mathcal{F}_1$. 
\end{proof}

\subsection{Proof of Theorem~\ref{thm:big_one}}\label{app:proof}
\begin{proof} 
This proof is laid out in three parts.

\noindent\underline{Part I}\\
First, note that if the parameterizations are one-to-one, then if one parameterization is shown to be identifiable, the others are identifiable as well. 
The density of the univariate generalized hyperbolic distribution is
\begin{equation}\label{univariate ghy}
f( x \mid\vectheta ) = 
\left[ \frac{ \omega + (x-\mu)^2/\sigma^2 }{ \omega+ \beta^2/ \sigma^{2} } \right]^{(\lambda-{1}/{2})/2}
\frac{  K_{\lambda - {p}/{2}}\Big(\sqrt{\big[ \omega+ \beta^2/\sigma^2 \big]\big[\omega +(x-\mu)^2/\sigma^2 \big]}\Big)}{ \sqrt{ 2\pi \sigma^2 }  K_{\lambda}\left( \omega\right)\exp\big\{-\left(x-\mu\right)\beta/\sigma^{2}\big\}},
\end{equation}
where $(\omega,\sigma,\mu,\beta,\lambda)\in\mathbb{R}_+^2\times\mathbb{R}^3$. After some algebraic manipulation, \eqref{univariate ghy} can be written
\begin{equation*} 
f( x \mid\vectheta ) = 
\left[ \frac{ 1 + (x-\mu)^2/ (\sqrt{\omega} \sigma)^2 }{ 1+ \beta^2/(\sqrt{\omega} \sigma)^2 } \right]^{(\lambda-{1}/{2})/2}
\frac{  K_{\lambda - {p}/{2}}\Big( \sqrt{ \big[ \omega+ \beta^2/\sigma^2 \big]/\sigma^2\big[ \omega\sigma^2 +(x-\mu)^2 \big]}\Big)}{ \sqrt{ 2\pi \sigma^2 }  K_{\lambda}\left( \omega\right)\exp\big\{-\left(x-\mu\right)\beta/\sigma^{2}\big\}}.
\end{equation*}
Let $\delta = \beta / \sigma^2$, $\alpha =  1/\sigma \times\sqrt{ \omega+ \beta^2/\sigma^2 }$ and $\kappa = \sigma \sqrt{\omega} $ , where $\alpha \ge | \delta | $. This implies that  $\sigma^2 = \kappa/\sqrt{\alpha^2-\delta^2}$, $\omega=  \kappa\sqrt{\alpha^2-\delta^2}$ and $\beta = \delta \kappa /\sqrt{\alpha^2-\delta^2}$, and another parameterization of the univariate generalized hyperbolic distribution emerges:
\begin{equation}\label{eqn:altpara}
f( x \mid\vectheta ) = 
\left[ \frac{ 1 + (x-\mu)^2/ \kappa^2 }{ 1+ \delta^2/(\alpha^2 - \delta^2 ) } \right]^{(\lambda-{1}/{2})/2}
\frac{  K_{\lambda - {p}/{2}}\Big( \alpha \sqrt{ \big[ \kappa^2  +(x-\mu)^2 \big]}\Big)}{ \sqrt{ 2\pi \sigma^2 }  K_{\lambda}\left( \kappa \sqrt{\alpha-\delta^2} \right)\exp\big\{-\left(x-\mu\right)\delta\big\}}.
\end{equation}

For large $z$, the Bessel function can approximated by  
\begin{equation*} 
K_{\lambda} (z) = \sqrt{ \frac{ \pi }{2 z} } e^{-z} \left[ 1+ O\left(\frac{1}{z}\right)\right],
\end{equation*} 
which yields 
\begin{equation*} \label{univariate ghy large x}
f( x \mid\vectheta ) \propto  
\left[ \omega + \frac{(x-\mu)^2}{\sigma^2} \right]^{\lambda/2-{3}/{4}}
\exp\left\{ - \sqrt{ \omega+ \beta^2/\sigma^2 } \frac{|x-\mu|}{\sigma}
 + \left(x-\mu\right)\beta/\sigma^{2} 
\right\}
\end{equation*}
or, if we use the alternative parameterization,
\begin{equation} \label{ density large x } 
f( x \mid\vectheta ) \propto  
\left[ 1 + \frac{(x-\mu)^2}{\kappa^2} \right]^{\lambda/2}\exp\left\{ - \alpha |x-\mu|+\delta \left(x-\mu\right)
\right\}.
\end{equation} 
The characteristic function for a normal mean-variance density  is 
\begin{equation*}
\varphi_{X}( t ) = \exp\{i t'\mu\} M_{W} \left( \beta t i -\frac{1}{2} \sigma^2 t^2 \; \big| \; \lambda, \omega \right),
\end{equation*}
and the moment generating function is
\begin{equation*}
M_{W} \left( u \right) 
 = \left[ \frac{\omega}{\omega -2u }  \right]^{ \frac{\lambda}{2} }
 \frac{ K_{\lambda} \left( \sqrt{ \omega(\omega-2u) } \right) }{   K_{\lambda} \left( \omega \right) }
  = \left[ 1 -2 \frac{u}{ \omega}  \right]^{- \frac{\lambda}{2} }
 \frac{ K_{\lambda} \left( \sqrt{ \omega(\omega-2u) } \right) }{   K_{\lambda} \left( \omega \right) }.
\end{equation*}
Now,
\begin{equation*}
\varphi_{X}( t ) = 
\exp\{ i  t'\mu\} 
\left[ 1 + \frac{ \sigma^2 t^2 -2 \beta t i  }{\omega}   \right]^{ -\frac{\lambda}{2} }
\frac{  K_{\lambda} \left( \sqrt{ \omega \left[\omega + ( \sigma^2 t^2 - 2\beta t i  )  \right] } \right)  }{  K_{\lambda} \left(  \omega \right) }.
\end{equation*}
Using the alternative parameterization in \eqref{eqn:altpara},
\begin{equation*}
\varphi_{X}( t ) = 
\exp\{ i  t'\mu \} 
\left[ 1 + \frac{ t^2 - 2 \delta t i }{\alpha^2 - \delta^2} \right]^{ -\frac{\lambda}{2} }
\frac{  K_{\lambda} \left( \sqrt{  \kappa^2 \left[  t^2 - 2 \delta t i + \alpha^2 - \delta^2   \right] } \right)  }{  K_{\lambda} \left(  \kappa \sqrt{\alpha^2 - \delta^2} \right)}.
\end{equation*}
For large $t$, the characteristic function is 
\begin{equation} 
\varphi_{X}( t ) \propto
\exp\{ i  t'\mu\} 
\left[ 1 + \frac{ t^2 - 2 \delta t i }{\alpha^2 - \delta^2}   \right]^{ -\frac{\lambda}{2} }
\exp\{-[ t \kappa +   O( 1 )]\}.
\end{equation}

\noindent\underline{Part II}\\
Let $\mathcal{G}_{(\alpha, \delta)}$ be the set of hyperbolic distributions with parameters $(\alpha, \delta)$ as well as $(\kappa, \mu, \lambda) \in (0,\infty) \times \mathbb{R}^2 $. Using the characteristic function, we will show that $\mathcal{G}_{(\alpha, \delta)}$ is identifiable. To begin, set $\mu=0$ and define $\gamma = (\kappa, \lambda)$ to get the characteristic function 
\begin{equation*}
\varphi_{\gamma}( t ) = 
\left[ 1 + \frac{ t^2 - 2 \delta t i }{\alpha^2 - \delta^2} \right]^{ -\frac{\lambda}{2} }
\frac{  K_{\lambda} \left( \sqrt{  \kappa^2 \left[  t^2 - 2 \delta t i + \alpha^2 - \delta^2   \right] } \right)  }{  K_{\lambda} \left(  \kappa \sqrt{\alpha^2 - \delta^2} \right) }.
\end{equation*}
For large $t$,
\begin{equation} 
\varphi_{\gamma}( t ) \propto
\left[ 1 + \frac{ t^2 - 2 \delta t i }{\alpha^2 - \delta^2}   \right]^{ -\frac{\lambda}{2} }
\exp\left\{  - \left[ t \kappa +   O( 1 ) \right]  \right\}.
\end{equation}
Now, define a total ordering on $\gamma$ as follows: $\gamma_1  \preceq \gamma_2$ if $ \kappa_2  > \kappa_1 $ or  if $ \kappa_2  =  \kappa_1 $ and $ \lambda_2  > \lambda_1 $. It follows that
\begin{equation}  \label{characteristic fn gamma} 
\lim_{t \rightarrow \infty} \frac{ \varphi_{\gamma_2} (t) } {  \varphi_{\gamma_1}(t) } =0.
\end{equation}

Now, with this ground work having been laid down, the proof follows similarly to Theorem~1 in \cite{holzmann2006}.  Starting with a dependence relation like in \eqref{ nonidentify }, but assuming that $ f  \in  \mathcal{G}_{(\alpha, \delta)}$ and applying the Fourier transform, we obtain
\begin{equation} 
\sum_{g=1}^G  \tau_g \exp\{ i  t \mu_g \}   \varphi_{\gamma_g} ( x ) = 0.
\end{equation}
Without loss of generality, assume that $\gamma_1 \preceq \dots \preceq \gamma_G$ and let $m \ge 1$ such that $ \gamma_1 = \dots = \gamma_m \preceq \gamma_{m+1}  \preceq \dots \preceq \gamma_{G}$. Dividing by $\exp\{ i  t \mu_1 \} \varphi_{\gamma_1}(t)$ we get
\begin{equation}\label{eqn:26}
\tau_1 
+ \sum_{g=2}^m \tau_g \exp\{ i  t \left( \mu_g -\mu_1 \right) \} + \sum_{g=m+1}^G \tau_g \exp\{ i  t \left( \mu_g -\mu_1 \right) \}  \frac{ \varphi_{\gamma_g} ( t ) }{   \varphi_{\gamma_1} ( t )} = 0.
\end{equation}
From $\eqref{characteristic fn gamma}$, the third term of $\eqref{eqn:26}$ $\rightarrow 0$ as $t\rightarrow \infty$ and, from \citet[][p.~762]{holzmann2006}, the second term $\rightarrow 0$ as $t\rightarrow \infty$. Therefore, $\tau_1 =0$. This is the beginning of a straightforward inductive argument.\\

\noindent\underline{Part III}\\
We have shown that $\mathcal{G}_{(\alpha, \delta)}$ is identifiable. It remains to show that  $\mathcal{G}_{(\alpha_1, \delta_1)} \cap \mathcal{G}_{(\alpha_2, \delta_2)} = \varnothing$. Assume that an element of $\mathcal{G}_{(\alpha_2, \delta_2)}$ can be expressed by a linear combination of elements from $\mathcal{G}_{(\alpha_1, \delta_1)}$, i.e., that there exist some $\tau_g \in \mathbb{R}$, $g=1,\ldots,G$, such that
\begin{equation} \label{span}
f(x \mid \mu_2, \lambda_2, \kappa_2, \alpha_2, \delta_2)  = \sum_{g=1}^G \tau_g f( x \mid \mu_{1g}, \lambda_{1g}, \kappa_{1g}, \alpha_1, \delta_1).
\end{equation} 
Dividing by the term on the left 
and defining 
\begin{equation}
 c_{g2} =  \frac{ f( x \mid \mu_{1g}, \lambda_{1g}, \kappa_{1g}, \alpha_1, \delta_1  ) }{ f( x \mid \mu_2, \lambda_2, \kappa_2, \alpha_2, \delta_2 )}, 
\end{equation} 
we have that $ 1= \sum_{i=1} \tau_g  c_{g2}$. Using equation (\ref{ density large x }) for large $x$, we have that as $x \rightarrow \infty$ and if  $  (\alpha_1 -\delta_1 ) > (\alpha_2 -\delta_2 )$  then  $c_{g2} \rightarrow 0$, and if  $  (\alpha_1 -\delta_1 ) < (\alpha_2 -\delta_2 )$  then $c_{g2} \rightarrow \infty $, yielding a contraction. Therefore, the intersection is empty. However, if $  (\alpha_1 -\delta_1 ) = (\alpha_2 -\delta_2 )$ then there may exist a set of $\tau_g$'s such that equation (\ref{span}) holds.

We again use \eqref{ density large x } for large $x$ but as $x \rightarrow -\infty$ and if  $  (\alpha_1 + \delta_1 ) > (\alpha_2 + \delta_2 )$  then  $c_{g2} \rightarrow 0$, and if  $  (\alpha_1 +\delta_1 ) < (\alpha_2 + \delta_2 )$  then $c_{g2} \rightarrow \infty $, yielding a contraction; therefore, the intersection is empty. But if $(\alpha_1 +\delta_1 ) = (\alpha_2 + \delta_2 )$ there may exist a set of $\tau_g$'s such that equation (\ref{span}) holds.


Combining these two arguments to have a set of $\tau_g$'s such that \eqref{span} holds we require $(\alpha_1 +\delta_1 ) = (\alpha_2 + \delta_2 )$ and $(\alpha_1 - \delta_1 ) = (\alpha_2 - \delta_2 )$, which means the two sets are equal. If either condition does not hold, we have that $\mathcal{G}_{(\alpha_1, \delta_1)} \cap \mathcal{G}_{(\alpha_2, \delta_2)} = \varnothing$.

Accordingly, the univariate generalized hyperbolic distribution is identifiable because it is formed from the union of identifiable and disjoint sets $\mathcal{G}_{(\alpha, \delta)}$ with total ordering on $(\alpha, \delta)$, i.e., $(\alpha_1, \delta_1)  \preceq (\alpha_2, \delta_2)$ if $ (\alpha_1 -\delta_1 ) < (\alpha_2 -\delta_2 )$ and $(\alpha_1 + \delta_1 ) > (\alpha_2 + \delta_2 )$, and if $(\alpha_1 -\delta_1 ) = (\alpha_2 -\delta_2 )$ then $(\alpha_1, \delta_1 ) \preceq (\alpha_2, \delta_2)$ if $(\alpha_1 + \delta_1 ) > (\alpha_2 + \delta_2)$. 
\end{proof}

\section{The GPCM Family}\label{app:gpcm}
Each member of the GPCM family is a Gaussian mixture model with eigen-decomposed component covariances, i.e., $\matsig_g=\lambda_g\matd_g\mata_g\matd_g'$, where $\matd_g$ is the matrix of eigenvectors, $\mata_g$ is a diagonal matrix with entries proportional to the eigenvalues, and $\lambda_g$ is the associated constant of proportionality \citep{banfield93}. Imposing constraints on the constituent elements of this decomposition leads to a family of 14 models (Table~\ref{tab:mclust2}) known as the GPCM family \citep{celeux95}. The MCLUST family is a subset of 10 of these models, and the models marked with an asterisk in Table~\ref{tab:mclust2} are not part of the MCLUST family.
\begin{table*}[!h]
\caption{\label{tab:mclust2}Nomenclature, covariance structure, and number of free covariance parameters for each member of the GPCM family. Models marked with an asterisk are not part of the MCLUST family.} 
\begin{tabular*}{1.0\textwidth}{@{\extracolsep{\fill}}lllllr}
\hline
Model &  Volume & Shape & Orientation & $\matsig_g$ & Free covariance parameters\\
\hline
EII & Equal & Spherical & --& $\lambda\ident$ & 1\\
VII & Variable & Spherical & -- & $\lambda_g\ident$ & $G$\\
EEI & Equal & Equal & Axis-Aligned & $\lambda\mata$ & $p$\\
VEI & Variable & Equal & Axis-Aligned & $\lambda_g\mata$ & $p+G-1$\\
EVI & Equal & Variable & Axis-Aligned & $\lambda\mata_g$ & $pG-G+1$\\
VVI & Variable & Variable & Axis-Aligned & $\lambda_g\mata_g$ & $pG$\\
EEE & Equal & Equal & Equal & $\lambda\matd\mata\matd'$ & $p(p+1)/2$\\
EEV & Equal & Equal & Variable & $\lambda\matd_g\mata\matd_g'$ & $Gp(p+1)/2 - (G-1)p$\\
VEV & Variable & Equal & Variable & $\lambda_g\matd_g\mata\matd_g'$ & $Gp(p+1)/2 - (G-1)(p-1)$\\
VVV & Variable & Variable & Variable & $\lambda_g\matd_g\mata_g\matd_g'$ & $Gp(p+1)/2$\\
EVE$^*$ & Equal & Variable & Equal & $\lambda\matd\mata_g\matd'$ & $p(p+1)/2 +(G-1)(p-1)$\\
VVE$^*$ & Variable & Variable & Equal & $\lambda_g\matd\mata_g\matd'$ & $p(p+1)/2 + (G-1)p$\\
VEE$^*$ & Variable & Equal & Equal & $\lambda_g\matd\mata\matd'$ & $p(p+1)/2 + (G-1)$\\
EVV$^*$ & Equal & Variable & Variable & $\lambda\matd_g\mata_g\matd_g'$ & $Gp(p+1)/2-(G-1)$\\
\hline
\end{tabular*}
\end{table*}


\begin{thebibliography}{}

\bibitem[\protect\citeauthoryear{****}{****}{2013}]{smcnicholas13}
**** (2013).
\newblock Title blinded.
\newblock {\em Arxiv preprint\/}.

\bibitem[\protect\citeauthoryear{Aitken}{Aitken}{1926}]{aitken26}
Aitken, A.~C. (1926).
\newblock On bernoulli's numerical solution of algebraic equations.
\newblock {\em Proceedings of the Royal Society of Edinburgh\/}~{\em 46},
  289--305.

\bibitem[\protect\citeauthoryear{Andrews and McNicholas}{Andrews and
  McNicholas}{2011}]{andrews11b}
Andrews, J.~L. and P.~D. McNicholas (2011).
\newblock Extending mixtures of multivariate t-factor analyzers.
\newblock {\em Statistics and Computing\/}~{\em 21\/}(3), 361--373.

\bibitem[\protect\citeauthoryear{Andrews and McNicholas}{Andrews and
  McNicholas}{2012}]{andrews12}
Andrews, J.~L. and P.~D. McNicholas (2012).
\newblock Model-based clustering, classification, and discriminant analysis via
  mixtures of multivariate $t$-distributions.
\newblock {\em Statistics and Computing\/}~{\em 22\/}(5), 1021--1029.

\bibitem[\protect\citeauthoryear{Andrews and McNicholas}{Andrews and
  McNicholas}{2014}]{andrews14}
Andrews, J.~L. and P.~D. McNicholas (2014).
\newblock Variable selection for clustering and classification.
\newblock {\em Journal of Classification\/}~{\em 31\/}(2), 136--153.

\bibitem[\protect\citeauthoryear{Atienza, Garcia-Heras, and Mu\~{n}oz
  Pichardo}{Atienza et~al.}{2006}]{atienza2006}
Atienza, N., J.~Garcia-Heras, and J.~Mu\~{n}oz Pichardo (2006).
\newblock A new condition for identifiability of finite mixture distributions.
\newblock {\em Metrika\/}~{\em 63}, 215--Ð221.

\bibitem[\protect\citeauthoryear{Banfield and Raftery}{Banfield and
  Raftery}{1993}]{banfield93}
Banfield, J.~D. and A.~E. Raftery (1993).
\newblock Model-based {G}aussian and non-{G}aussian clustering.
\newblock {\em Biometrics\/}~{\em 49\/}(3), 803--821.

\bibitem[\protect\citeauthoryear{Baricz}{Baricz}{2010}]{baricz2010}
Baricz, A. (2010).
\newblock Tur‡n type inequalities for some probability density functions.
\newblock {\em Studia scientiarum mathematicarum Hungarica\/}~{\em 47},
  175--189.

\bibitem[\protect\citeauthoryear{Barndorff-Nielsen}{Barndorff-Nielsen}{1978}]{barndorff78}
Barndorff-Nielsen, O. (1978).
\newblock Hyperbolic distributions and distributions on hyperbolae.
\newblock {\em Scandinavian Journal of Statistics\/}~{\em 5}, 151--157.

\bibitem[\protect\citeauthoryear{Barndorff-Nielsen and
  Halgreen}{Barndorff-Nielsen and Halgreen}{1977}]{barndorff77}
Barndorff-Nielsen, O. and C.~Halgreen (1977).
\newblock Infinite divisibility of the hyperbolic and generalized inverse
  {G}aussian distributions.
\newblock {\em Z.\ Wahrscheinlichkeitstheorie Verw.\ Gebiete\/}~{\em 38},
  309--311.

\bibitem[\protect\citeauthoryear{Biernacki, Celeux, and Gold}{Biernacki
  et~al.}{2000}]{biernacki2000}
Biernacki, C., G.~Celeux, and E.~Gold (2000).
\newblock Assessing a mixture model for clustering with the integrated
  completed likelihood.
\newblock {\em IEEE Transactions on Pattern Analysis and Machine
  Intelligence\/}~{\em 22}, 719Ð725.

\bibitem[\protect\citeauthoryear{Bl{\ae}sild}{Bl{\ae}sild}{1978}]{blaesild78}
Bl{\ae}sild, P. (1978).
\newblock The shape of the generalized inverse {G}aussian and hyperbolic
  distributions.
\newblock Research Report~37, Department of Theoretical Statistics, Aarhus
  University, Denmark.

\bibitem[\protect\citeauthoryear{B\"{o}hning, Dietz, Schaub, Schlattmann, and
  Lindsay}{B\"{o}hning et~al.}{1994}]{bohning94}
B\"{o}hning, D., E.~Dietz, R.~Schaub, P.~Schlattmann, and B.~Lindsay (1994).
\newblock The distribution of the likelihood ratio for mixtures of densities
  from the one-parameter exponential family.
\newblock {\em Annals of the Institute of Statistical Mathematics\/}~{\em 46},
  373--388.

\bibitem[\protect\citeauthoryear{Browne and McNicholas}{Browne and
  McNicholas}{2013}]{browne13}
Browne, R.~P. and P.~D. McNicholas (2013).
\newblock {\em mixture: Mixture Models for Clustering and Classification}.
\newblock R package version 1.0.

\bibitem[\protect\citeauthoryear{Browne and McNicholas}{Browne and
  McNicholas}{2014}]{browne13b}
Browne, R.~P. and P.~D. McNicholas (2014).
\newblock Estimating common principal components in high dimensions.
\newblock {\em Advances in Data Analysis and Classification\/}~{\em 8\/}(2).
\newblock 217--226.

\bibitem[\protect\citeauthoryear{Browne, McNicholas, and Sparling}{Browne
  et~al.}{2012}]{browne11}
Browne, R.~P., P.~D. McNicholas, and M.~D. Sparling (2012).
\newblock Model-based learning using a mixture of mixtures of {G}aussian and
  uniform distributions.
\newblock {\em IEEE Transactions on Pattern Analysis and Machine
  Intelligence\/}~{\em 34\/}(4), 814--817.

\bibitem[\protect\citeauthoryear{Cabral, Lachos, and Prates}{Cabral
  et~al.}{2012}]{cabral2012}
Cabral, C. R.~B., V.~H. Lachos, and M.~O. Prates (2012).
\newblock Multivariate mixture modeling using skew-normal independent
  distributions.
\newblock {\em Computational Statistics and Data Analysis\/}~{\em 56\/}(1), 126
  -- 142.

\bibitem[\protect\citeauthoryear{Campbell and Mahon}{Campbell and
  Mahon}{1974}]{campbell74}
Campbell, N.~A. and R.~J. Mahon (1974).
\newblock A multivariate study of variation in two species of rock crab of
  genus leptograpsus.
\newblock {\em Australian Journal of Zoology\/}~{\em 22}, 417–--425.

\bibitem[\protect\citeauthoryear{Celeux and Govaert}{Celeux and
  Govaert}{1995}]{celeux95}
Celeux, G. and G.~Govaert (1995).
\newblock Gaussian parsimonious clustering models.
\newblock {\em Pattern Recognition\/}~{\em 28\/}(5), 781--793.

\bibitem[\protect\citeauthoryear{Chen and Tan}{Chen and Tan}{2009}]{chen09}
Chen, J. and X.~Tan (2009).
\newblock Inference for multivariate normal mixtures.
\newblock {\em Journal of Multivariate Analysis\/}~{\em 100\/}(7), 1367--1383.

\bibitem[\protect\citeauthoryear{Dasgupta and Raftery}{Dasgupta and
  Raftery}{1998}]{dasgupta98}
Dasgupta, A. and A.~E. Raftery (1998).
\newblock Detecting features in spatial point processes with clutter via
  model-based clustering.
\newblock {\em Journal of the American Statistical Association\/}~{\em 93},
  294--302.

\bibitem[\protect\citeauthoryear{de~Leeuw and Lange}{de~Leeuw and
  Lange}{2009}]{deleeuw2009}
de~Leeuw, J. and K.~Lange (2009).
\newblock Sharp quadratic majorization in one dimension.
\newblock {\em Computational Statistics and Data Analysis\/}~{\em 53\/}(7),
  2471--2484.

\bibitem[\protect\citeauthoryear{Dean, Murphy, and Downey}{Dean
  et~al.}{2006}]{dean06}
Dean, N., T.~B. Murphy, and G.~Downey (2006).
\newblock Using unlabelled data to update classification rules with
  applications in food authenticity studies.
\newblock {\em Journal of the Royal Statistical Society: Series C\/}~{\em
  55\/}(1), 1--14.

\bibitem[\protect\citeauthoryear{Dempster, Laird, and Rubin}{Dempster
  et~al.}{1977}]{dempster77}
Dempster, A.~P., N.~M. Laird, and D.~B. Rubin (1977).
\newblock Maximum likelihood from incomplete data via the {EM} algorithm.
\newblock {\em Journal of the Royal Statistical Society: Series B\/}~{\em
  39\/}(1), 1--38.

\bibitem[\protect\citeauthoryear{Fraley and Raftery}{Fraley and
  Raftery}{2002}]{fraley02a}
Fraley, C. and A.~E. Raftery (2002).
\newblock Model-based clustering, discriminant analysis, and density
  estimation.
\newblock {\em Journal of the American Statistical Association\/}~{\em
  97\/}(458), 611--631.

\bibitem[\protect\citeauthoryear{Fraley, Raftery, and Scrucca}{Fraley
  et~al.}{2013}]{fraley13}
Fraley, C., A.~E. Raftery, and L.~Scrucca (2013).
\newblock {mclust}: Normal mixture modeling for model-based clustering,
  classification, and density estimation.
\newblock R package version 4.2.

\bibitem[\protect\citeauthoryear{Franczak, Browne, and McNicholas}{Franczak
  et~al.}{2014}]{franczak14}
Franczak, B.~C., R.~P. Browne, and P.~D. McNicholas (2014).
\newblock Mixtures of shifted asymmetric {L}aplace distributions.
\newblock {\em IEEE Transactions on Pattern Analysis and Machine
  Intelligence\/}~{\em 36\/}(6), 1149--1157.

\bibitem[\protect\citeauthoryear{Good}{Good}{1953}]{good53}
Good, I, J. (1953).
\newblock The population frequencies of species and the estimation of
  population parameters.
\newblock {\em Biometrika\/}~{\em 40}, 237--260.

\bibitem[\protect\citeauthoryear{Halgreen}{Halgreen}{1979}]{halgreen79}
Halgreen, C. (1979).
\newblock Self-decomposibility of the generalized inverse {G}aussian and
  hyperbolic distributions.
\newblock {\em Z.\ Wahrscheinlichkeitstheorie Verw.\ Gebiete\/}~{\em 47},
  13--18.

\bibitem[\protect\citeauthoryear{Handcock, Raftery, and Tantrum}{Handcock
  et~al.}{2007}]{handcock07}
Handcock, M.~S., A.~E. Raftery, and J.~M. Tantrum (2007).
\newblock Model-based clustering for social networks.
\newblock {\em Journal of the Royal Statistical Society: Series~A\/}~{\em
  170\/}(2), 301--354.

\bibitem[\protect\citeauthoryear{Hastie and Tibshirani}{Hastie and
  Tibshirani}{1996}]{hastie96}
Hastie, T. and R.~Tibshirani (1996).
\newblock Discriminant analysis by {G}aussian mixtures.
\newblock {\em Journal of the Royal Statistical Society: Series B\/}~{\em
  58\/}(1), 155--176.

\bibitem[\protect\citeauthoryear{Hathaway}{Hathaway}{1985}]{hathaway85}
Hathaway, R.~J. (1985).
\newblock A constrained formulation of maximum-likelihood estimation for normal
  mixture distributions.
\newblock {\em The Annals of Statistics\/}~{\em 13\/}(2), 795--800.

\bibitem[\protect\citeauthoryear{Holzmann, Munk, and Gneiting}{Holzmann
  et~al.}{2006}]{holzmann2006}
Holzmann, H., A.~Munk, and T.~Gneiting (2006).
\newblock Identifiability of finite mixtures of elliptical distributions.
\newblock {\em Scandinavian Journal of Statistics\/}~{\em 33}, 753--763.

\bibitem[\protect\citeauthoryear{Hubert and Arabie}{Hubert and
  Arabie}{1985}]{hubert85}
Hubert, L. and P.~Arabie (1985).
\newblock Comparing partitions.
\newblock {\em Journal of Classification\/}~{\em 2\/}(1), 193--218.

\bibitem[\protect\citeauthoryear{J{\o}rgensen}{J{\o}rgensen}{1982}]{jorgensen82}
J{\o}rgensen, B. (1982).
\newblock {\em Statistical Properties of the Generalized Inverse Gaussian
  Distribution}.
\newblock New York: Springer-Verlag.

\bibitem[\protect\citeauthoryear{Karlis and Meligkotsidou}{Karlis and
  Meligkotsidou}{2007}]{karlis07}
Karlis, D. and L.~Meligkotsidou (2007).
\newblock Finite mixtures of multivariate poisson distributions with
  application.
\newblock {\em Journal of Statistical Planning and Inference\/}~{\em 137\/}(6),
  1942--1960.

\bibitem[\protect\citeauthoryear{Kass and Raftery}{Kass and
  Raftery}{1995}]{kass95}
Kass, R.~E. and A.~E. Raftery (1995).
\newblock Bayes factors.
\newblock {\em Journal of the American Statistical Association\/}~{\em
  90\/}(430), 773--795.

\bibitem[\protect\citeauthoryear{Kass and Wasserman}{Kass and
  Wasserman}{1995}]{kass95b}
Kass, R.~E. and L.~Wasserman (1995).
\newblock A reference {B}ayesian test for nested hypotheses and its
  relationship to the {S}chwarz criterion.
\newblock {\em Journal of the American Statistical Association\/}~{\em
  90\/}(431), 928--934.

\bibitem[\protect\citeauthoryear{Kent}{Kent}{1983}]{kent1983}
Kent, J.~T. (1983).
\newblock Identifiability of finite mixtures for directional data.
\newblock {\em The Annals of Statistics\/}~{\em 11}, 984--988.

\bibitem[\protect\citeauthoryear{Kotz, Kozubowski, and Podgorski}{Kotz
  et~al.}{2001}]{kotz2001}
Kotz, S., T.~J. Kozubowski, and K.~Podgorski (2001).
\newblock {\em The Laplace Distribution and Generalizations: A Revisit with
  Applications to Communications, Economics, Engineering, and Finance}.
\newblock Boston: Birkhauser.

\bibitem[\protect\citeauthoryear{Lee and McLachlan}{Lee and
  McLachlan}{2014}]{lee14}
Lee, S.\ X.\ and G.~J. McLachlan (2014).
\newblock Finite mixtures of multivariate skew t-distributions: some recent and new
results.
\newblock {\em Statistics and Computing\/}~{\em 24\/}(2), 181--202.

\bibitem[\protect\citeauthoryear{Lin}{Lin}{2009}]{lin09}
Lin, T.-I. (2009).
\newblock Maximum likelihood estimation for multivariate skew normal mixture
  models.
\newblock {\em Journal of Multivariate Analysis\/}~{\em 100}, 257--265.

\bibitem[\protect\citeauthoryear{Lin}{Lin}{2010}]{lin10}
Lin, T.-I. (2010).
\newblock Robust mixture modeling using multivariate skew t distributions.
\newblock {\em Statistics and Computing\/}~{\em 20\/}(3), 343--356.

\bibitem[\protect\citeauthoryear{Lin, McNicholas, and Hsiu}{Lin et~al.}{2014}]{lin14}
Lin, T.-I., McNicholas, P.~D., and Hsiu, J.~H. (2014) (2009).
\newblock Capturing patterns via parsimonious t mixture models
\newblock {\em Statistics and Probability Letters\/}~{\em 88}, 80--87.

\bibitem[\protect\citeauthoryear{Lindsay}{Lindsay}{1995}]{lindsay95}
Lindsay, B.~G. (1995).
\newblock Mixture models: Theory, geometry and applications.
\newblock In {\em NSF-CBMS Regional Conference Series in Probability and
  Statistics}, Volume~5. California: Institute of Mathematical Statistics:
  Hayward.

\bibitem[\protect\citeauthoryear{McLachlan}{McLachlan}{1992}]{mclachlan92}
McLachlan, G.~J. (1992).
\newblock {\em Discriminant Analysis and Statistical Pattern Recognition}.
\newblock New Jersey: John Wiley \& Sons.

\bibitem[\protect\citeauthoryear{McNeil, Frey, and Embrechts}{McNeil
  et~al.}{2005}]{mcneil2005}
McNeil, A.~J., R.~Frey, and P.~Embrechts (2005).
\newblock {\em Quantitative Risk Management: Concepts, Techniques and Tools}.
\newblock Princeton University Press.

\bibitem[\protect\citeauthoryear{McNicholas}{McNicholas}{2010}]{mcnicholas10c}
McNicholas, P.~D. (2010).
\newblock Model-based classification using latent {G}aussian mixture models.
\newblock {\em Journal of Statistical Planning and Inference\/}~{\em 140\/}(5),
  1175--1181.

\bibitem[\protect\citeauthoryear{McNicholas, Murphy, McDaid, and Frost}{McNicholas et al.}{2010}]{mcnicholas10a}
McNicholas, P.~D., T.~B. Murphy, A.~F. McDaid, and D.~Frost. (2010)
\newblock Serial and parallel implementations of model-based clustering via
  parsimonious {G}aussian mixture models.
\newblock {\em Computational Statistics and Data Analysis\/}~{\em 54\/}(3), 
711--723.

\bibitem[\protect\citeauthoryear{Melnykov}{Melnykov}{2013}]{melnykov2013}
Melnykov, V. (2013).
\newblock Journal of multivariate analysis.
\newblock {\em Scandinavian Journal of Statistics\/}~{\em 122}, 175--189.

\bibitem[\protect\citeauthoryear{Murray, Browne, and McNicholas}{Murray
  et~al.}{2014a}]{murray14a}
Murray, P.~M., R.~B. Browne, and P.~D. McNicholas (2014a).
\newblock Mixtures of skew-t factor analyzers.
\newblock {\em Computational Statistics and Data Analysis\/}~{\em 77},
  326--335.

\bibitem[\protect\citeauthoryear{Murray, McNicholas, and Browne}{Murray
  et~al.}{2014b}]{murray14b}
Murray, P.~M., P.~D. McNicholas, and R.~B. Browne (2014b).
\newblock A mixture of common skew-$t$ factor analyzers.
\newblock {\em Stat\/}~{\em 3\/}(1), 68--82.

\bibitem[\protect\citeauthoryear{Peel and McLachlan}{Peel and
  McLachlan}{2000}]{peel00}
Peel, D. and G.~J. McLachlan (2000).
\newblock Robust mixture modelling using the t distribution.
\newblock {\em Statistics and Computing\/}~{\em 10\/}(4), 339--348.

\bibitem[\protect\citeauthoryear{Prates, Cabral, and Lachos}{Prates
  et~al.}{2013}]{prates2013}
Prates, M.~O., C.~B. Cabral, and V.~H. Lachos (2013).
\newblock {mixsmsn}: Fitting finite mixture of scale mixture of skew-normal
  distributions.

\bibitem[\protect\citeauthoryear{{R Core Team}}{{R Core Team}}{2013}]{r2013}
{R Core Team} (2013).
\newblock {\em R: A Language and Environment for Statistical Computing}.
\newblock Vienna, Austria: R Foundation for Statistical Computing.

\bibitem[\protect\citeauthoryear{Raftery and Dean}{Raftery and
  Dean}{2006}]{raftery06}
Raftery, A.~E. and N.~Dean (2006).
\newblock Variable selection for model-based clustering.
\newblock {\em Journal of the American Statistical Association\/}~{\em
  101\/}(473), 168--178.

\bibitem[\protect\citeauthoryear{Rand}{Rand}{1971}]{rand71}
Rand, W.~M. (1971).
\newblock Objective criteria for the evaluation of clustering methods.
\newblock {\em Journal of the American Statistical Association\/}~{\em
  66\/}(336), 846--850.

\bibitem[\protect\citeauthoryear{Redner and Walker}{Redner and
  Walker}{1984}]{render84}
Redner, R.~A. and H.~F. Walker (1984).
\newblock Mixture densities, maximum likelihood and the {EM} algorithm.
\newblock {\em SIAM Review\/}~{\em 26}, 195--239.

\bibitem[\protect\citeauthoryear{Schwarz}{Schwarz}{1978}]{schwarz1978}
Schwarz, G. (1978).
\newblock Estimating the dimension of a model.
\newblock {\em Annals of Statistics\/}~{\em 6}, 461--464.

\bibitem[\protect\citeauthoryear{Stephens}{Stephens}{2000}]{stephens00b}
Stephens, M. (2000).
\newblock Dealing with label switching in mixture models.
\newblock {\em Journal of the Royal Statistical Society Society:
  Series~B\/}~{\em 62\/}(4), 795--809.

\bibitem[\protect\citeauthoryear{Venables and Ripley}{Venables and
  Ripley}{1999}]{venables99}
Venables, W.~N. and B.~D. Ripley (1999).
\newblock {\em Modern Applied Statistics with S-PLUS}.
\newblock Springer.

\bibitem[\protect\citeauthoryear{Vrbik and McNicholas}{Vrbik and
  McNicholas}{2012}]{vrbik12}
Vrbik, I. and P.~D. McNicholas (2012).
\newblock Analytic calculations for the {EM} algorithm for multivariate
  skew-mixture models.
\newblock {\em Statistics and Probability Letters\/}~{\em 82\/}(6), 1169--1174.

\bibitem[\protect\citeauthoryear{Vrbik and McNicholas}{Vrbik and
  McNicholas}{2014}]{vrbik14}
Vrbik, I. and P.~D. McNicholas (2014).
\newblock Parsimonious skew mixture models for model-based clustering and
  classification.
\newblock {\em Computational Statistics and Data Analysis\/}~{\em 71},
  196--210.

\bibitem[\protect\citeauthoryear{Watson}{Watson}{1944}]{watson1944}
Watson, G.~N. (1944).
\newblock {\em A Treatise on the Theory of Bessel Functions}.
\newblock Cambridge: Cambridge University Press.

\bibitem[\protect\citeauthoryear{Yakowitz and Spragins}{Yakowitz and
  Spragins}{1968}]{yakowitz1968}
Yakowitz, S.~J. and J.~Spragins (1968).
\newblock On the identifiability of finite mixtures.
\newblock {\em Ann. Math. Statist.\/}~{\em 39}, 209--214.

\end{thebibliography}

\end{document}